\documentclass[11pt,a4paper]{article}

\usepackage{amsfonts,amsmath,amsthm,amssymb,authblk,graphicx,multirow,rotating,dsfont,color,natbib,enumerate,color,setspace,enumitem,pdflscape,xr,siunitx,hyperref,url,siunitx,caption}
\usepackage[utf8]{inputenc}
\usepackage[english]{babel}
\usepackage{tikz}
\usetikzlibrary{arrows,chains,matrix,positioning,scopes}

\usepackage{times}
\usepackage{bm}
\usepackage{natbib,xr,url}
\usepackage[plain,noend]{algorithm2e}
\usepackage{color,amsfonts,amssymb,dsfont,graphicx}

\newtheorem{proposition}{Proposition}
\newtheorem{definition}{Definition}

\newtheorem{assumption}{Assumption}

\def\1{{{\mbox{$\mathds{1}$}}}}

\def\Z{\mathcal{Z}}
\def\U{\mathcal{U}}

\def\P{\text{pr}}
\def\Pt{\text{pr}_\theta}
\def\Ptp{\text{pr}_{\theta^{[s]}}}
\def\Pv {\tilde {\text{pr}}_{\tau,\theta}}
\def\Pvtau {\tilde {\text{pr}}_{\tau}}
\def\Ev{\tilde {E}_{\tau,\theta}}
\def\Evtau{\tilde {E}_{\tau}}
\newcommand{\KL}[2]{\textsc{KL}\{#1 \| #2\}}
\newcommand{\ICL}{\textsc{icl}}
\newcommand{\argmax}{\mathop{\text{\rm Argmax}}}
\newcommand{\argmin}{\mathop{\text{\rm Argmin}}}

\title{A  semiparametric extension of the  stochastic block model for longitudinal networks}
\author{Catherine Matias, Tabea Rebafka and Fanny Villers \\
Sorbonne Universités, Université Pierre  et Marie Curie, Université Paris
Diderot, Centre National de la Recherche Scientifique, 
Laboratoire de  Probabilités   et    Modèles   Aléatoires, 4 place Jussieu, 75252 PARIS Cedex 05, FRANCE. \\
\texttt{\{catherine.matias,tabea.rebafka,fanny.villers\}@upmc.fr}
}
\date{}

\begin{document}
\DeclareGraphicsExtensions{.pdf, .jpg, .jpeg, .png, .gif,.eps}

\thispagestyle{empty}
\maketitle

\begin{abstract}
To model  recurrent interaction events in continuous time, an extension of the stochastic block model is proposed where every individual belongs to a latent group and interactions between two individuals follow a conditional 
inhomogeneous Poisson process with intensity driven by the individuals' latent groups. 
The model is shown to be identifiable and its estimation is based on a semiparametric  variational expectation-maximization algorithm. Two
 versions of the method are developed, using either a nonparametric histogram approach (with an adaptive choice of the partition
size) or kernel intensity estimators.   The number of latent groups can be selected by   an integrated classification likelihood criterion. 
Finally, we demonstrate the performance of our procedure on 
synthetic experiments, analyse two datasets to illustrate the utility  of our approach and comment on competing methods.  
 \end{abstract}


\section{Introduction}

The past  few years  have seen  a large increase  in the  interest for
modelling dynamic interactions between individuals. 
Continuous-time information on interactions is often available,
as e.g. email exchanges between employees in a company
~\citep{Enron} or
face-to-face contacts between individuals measured by sensors~\citep{Stehle}, but most models use discrete time.  
Commonly, data are aggregated 
on predefined 
time intervals  to obtain a sequence of snapshots of random graphs. Besides the 
 loss of information induced by aggregation, 
 the specific choice of the time intervals has a direct impact on the results, most often overlooked.  
 Thus,  developing   continuous-time models   
 -- either called \emph{longitudinal
  networks, interaction  event data,  link streams}  or \emph{temporal
  networks} -- is an important research issue.  

Statistical methods for longitudinal networks form a huge corpus, especially in social sciences and we are not  exhaustive here,
see~\cite{Holme_review} for a more complete review. 
It is natural to model temporal event data by  stochastic point processes.  
An important line of research  involves
 continuous-time Markov processes with seminal works on dyad-independent
models \citep{Wasserman_80a,Wasserman_80b} 
or the so-called stochastic actor
oriented models~\citep[e.g.][]{Snijders_vDuijn,snijders2010}. 
In these works interactions  last during some time.
In contrast, here we focus on instantaneous interactions  identified with  time points. 
Furthermore, we are concerned with modelling  dependencies 
between interactions  of pairs of individuals.

The analysis of event data is an old area in statistics~\citep[see e.g.][]{Andersen_book}.
Generally,  the number of interactions of all pairs $(i,j)$ of individuals up to time $t$ are modelled by a multivariate counting process $N(t)=(N_{i,j}(t))_{(i,j)}$. 
\cite{Butts_08} considers time-stamped interactions 
marked by a label representing a behavioural event. 
His model is an  instance of Cox's  multiplicative hazard
model with time-dependent covariates and constant baseline function. 
In the same vein, \cite{Vu_etal11} propose a regression-based modelling of the intensity of  non recurrent interaction
events. They consider two different frameworks: Cox's multiplicative and Aalen's additive hazard rates. 
\cite{Perry_Wolfe_JRSSB} propose another variant of Cox's multiplicative intensity model for recurrent interaction events
where the baseline function is specific to each individual. 
In the above mentioned works a set of statistics is chosen by the user 
that potentially modulate the interactions.
As in any regression framework, the choice of these statistics raises some
issues: increasing their number 
may lead 
to a high-dimensional problem, and interpretation of the results might
be blurred by their possible correlations.

The approaches by~\citeauthor{Butts_08}, \citeauthor{Vu_etal11}, \citeauthor{Perry_Wolfe_JRSSB} and others are based on conditional Poisson
processes characterized  by random  intensities, also known  as doubly
stochastic Poisson processes or Cox processes. 
An instance of the conditional Poisson process is the Hawkes
process, which is a collection of point processes with some background
rate, where each event adds a  nonnegative impulse to the intensity of
all other processes.  
Introducing reciprocating Hawkes to parameterize  edges \cite{Blundel2012}  extend  the Infinite  Relational Model making  all events   co-dependent over time. 
\cite{Cho_etal14} develop  a model for spatial-temporal  networks with
missing information,  based on Hawkes processes  for temporal dynamics
combined with a Gaussian mixture for the spatial dynamics. 
Similarly,~\cite{Linderman14} combine temporal  Hawkes processes with latent distance models for implicit networks that are not directly  observed. 
We also mention the existence of  models associating point processes  with single nodes rather than 
  pairs, see e.g.~\cite{Fox2016} and the references therein.

Clustering individuals based on interaction data is a well-established technique to take into account the
intrinsic heterogeneity and summarize information.   
For discrete-time sequences of graphs, recent
approaches propose  generalizations of the  stochastic block model  to a dynamic 
context~\citep{Yang_etal_ML11,Xu_Hero_IEEE,Corneli,Matias_Miele}. Stochastic
block models posit that all individuals  belong  to 
one out of finitely many groups and given these groups all pairs of interactions are
  independent. We stress  that stochastic block models induce
  more     general     clusterings    than     community     detection
  algorithms. 
  Indeed, clusters are  not necessarily characterized by  intense  within-group   interaction    and  low interaction  frequency   towards other groups. 
Another attempt to use stochastic block models 
for interaction events 
appears in~\cite{DuBois_Aistat} generalizing 
the approach of~\cite{Butts_08} by adding discrete latent variables on the individuals.

We introduce a semiparametric stochastic block model for recurrent interaction events in continuous time, to which we refer as the Poisson process stochastic block model. 
Interactions are modelled by conditional inhomogeneous
Poisson processes,  whose intensities only depend on the latent groups of the interacting individuals. 
In contrast to many other works, we do not rely on a parametric model  where intensities are modulated by predefined network statistics, but intensities are modelled and estimated in a nonparametric way. 
The model parameters are shown to be identifiable.
Our estimation and clustering approach is a semiparametric version of the variational expectation-maximization  algorithm, where the maximization step is replaced  by nonparametric estimators of the intensities.
Semiparametric       generalizations       of      the       classical
expectation-maximization   algorithm  have   been  proposed   in  many
different contexts; see 
e.g. \cite{Bohning,Bordes,Robin_kerfdr} for semiparametric mixtures or \cite{Dannemann} for a semiparametric hidden Markov
model. However,  we are  not aware of  other attempts  to incorporate
nonparametric  estimates  in  a   variational  approximation algorithm. 
We propose two different estimations  of the nonparametric part of the
model: a histogram approach using \cite{patricia} where the partition size
is adaptively chosen and a kernel estimator based on~\cite{Ramlau}.
 With the histogram approach,  an integrated classification likelihood
 criterion  is  proposed  to  select   the  number  of  latent  groups.  
Synthetic experiments and the  analysis of two    datasets illustrate the strengths and    weaknesses of our approach.
The R code is available in the R package \texttt{ppsbm}.

\section{A semiparametric Poisson process stochastic block model}
\subsection{Model}
\label{sec:notation}
We consider the pairwise interactions of $n$ individuals during some time interval $[0,T]$. 
For  notational convenience  we  restrict  our attention  to
directed interactions  without self-interactions. 
The undirected case and  self-interactions are handled similarly. 
The set of all pairs of individuals (i.e. the set of all possible dyads in the graph) is denoted 
$$
\mathcal R=\{(i,j) : i,j=1,\dots,n;\ i\neq j\},
$$
whose cardinality is $r=n(n-1)$. 
We observe the interactions during the time interval $[0,T]$, that is
$$
\mathcal O=\left\{(t_m,i_m,j_m),m=1,\dots,M\right\},
$$ 
where 
$(t_m,i_m,j_m)\in[0,T]\times \mathcal R$ corresponds to the event that
 a (directed) interaction from the $i_m$th individual 
to the $j_m$th individual occurs at time $t_m$. The total (random) number of events 
is $M$.
We assume that  $0<t_1<\dots< t_M<T$, i.e. there is at most one event at a time.

Every individual is assumed to belong to one out of $Q$ groups and the relation between two individuals, that is
the  way  they  interact  with  another,  is  driven  by  their  group
membership.  Let   $Z_1,\dots,Z_n$  be  independent   and  identically
distributed (latent) random variables taking values in $\{1,\dots,Q\}$
with non zero probabilities $ \P(Z_1=q) =\pi_q\ (q=1,\dots,Q)$. 
For the moment,  $Q$ is considered to be fixed and known.  
When no confusion occurs, we also use the notation $Z_i=(Z^{i,1},\dots,Z^{i,Q})$ with $Z^{i,q}\in\{0,1\}$ such that $Z_i$ has multinomial distribution $\mathcal M(1,\pi)$ with $\pi=(\pi_1,\dots,\pi_Q)$.

For every pair $(i,j)\in\mathcal R$ the interactions of individual $i$ to
$j$ conditional on the latent groups  $Z_1,\dots,Z_n$ are modelled by a conditional inhomogeneous Poisson 
process $N_{i,j}(\cdot)$  on $[0,T]$ with intensity  depending only on $Z_i$ and $Z_j$. 
For   $ q, l=1,\dots, Q$ and $(i,j)\in\mathcal R$  the conditional intensity of  $N_{i,j}(\cdot)$ given that $Z_i=q$
and $Z_j=l$ is $\alpha^{(q,l)}(\cdot)$ with corresponding cumulative intensity   
$$
A^{(q,l)} (t)= \int_0^t \alpha^{(q,l)}(u)du, \quad (0\le t \le T). 
$$
The set of observations $\mathcal O$ is a realization of  the multivariate counting process
$\{N_{i,j}(\cdot)\}_{(i,j)\in \mathcal{R}} $ with conditional intensity process $\{\alpha^{(Z_i,Z_j)}(\cdot)\}_{(i,j)\in
  \mathcal{R}} $. 
The process $N_{i,j}$ is not a Poisson process, but a  counting process with intensity $\sum_{q=1}^Q\sum_{l=1}^Q\pi_q\pi_l
  \alpha^{(q,l)}$. 
We denote $\theta =(\pi,\alpha)$ the infinite-dimensional parameter of
our model and $\Pt$ the  Poisson process stochastic block model distribution of the multivariate
counting process $\{N_{i,j}(\cdot)\}_{(i,j)\in \mathcal{R}}$. An 
extension of this model that specifically accounts for sparse interactions
  processes  is given  in the  Supplementary  Material.

\subsection{Identifiability}
\label{sec:ident}

Identifiability of the parameter $\theta$ corresponds to
  injectivity of the mapping $\theta  \mapsto \Pt$ and may be obtained
  at best up to label switching, as defined below.
We denote $\mathfrak{S}_Q$ the set of permutations of $\{1,\dots, Q\}$. 

\begin{definition}[Identifiability up to label switching]
  The parameter $\theta=(\pi, \alpha)$ of a Poisson process stochastic block model is identifiable on $[0,T]$ up to label switching if 
  for   all   $\theta$   and   $\tilde\theta$   such   that   $\Pt   =
  \text{pr}_{\tilde \theta}$ there exists a permutation $\sigma \in \mathfrak{S}_Q$ such that 
\begin{equation*}
 \pi_q =\tilde \pi_{\sigma(q)} , \quad
\alpha^{(q,l)} =  \tilde \alpha^{(\sigma(q),\sigma(l))}  \text{ almost
  everywhere on } [0,T], \quad (q, l=1,\dots, Q).
\end{equation*}
\end{definition}

Identifiability (up to label switching) is ensured in the very general setting where the intensities $\alpha^{(q,l)}$
are not equal almost everywhere, that is, they may be identical on at most  subsets of $[0,T]$.

\begin{assumption}
  \label{hyp:ident}
In the directed setup (resp. undirected), the set of intensities 
$\{\alpha^{(q,l)}\}_{q, l=1,\dots, Q}$ contains exactly $Q^2$ (resp. $Q(Q+1)/2$) distinct functions. 
\end{assumption}

\begin{proposition}
\label{prop:ident}
  Under        Assumption~\ref{hyp:ident},        the        parameter
  $\theta=(\pi,\alpha)$  is  identifiable  on  $[0,T]$ up  to  label
  switching  from the Poisson process stochastic block model 
  distribution $\Pt$, as  soon as $n\ge 3$.
\end{proposition}

Assumption~\ref{hyp:ident} is similar to the hypothesis from Theorem 12 in~\cite{AMR_JSPI} that to our knowledge is the only
 identifiability result for weighted stochastic block models.  
 The question whether the necessary condition that any two rows (or any two columns) of the parameter matrix
  $\alpha$ are distinct is a sufficient condition for identifiability is not yet answered 
even in the simpler binary case. 
  In the binary stochastic block model, the results in~\cite{AMR_AoS,AMR_JSPI}
  establish \emph{generic} identifiability, which
  means identifiability except on a subset of parameters with Lebesgue measure zero, without specifying this subset. For the directed and binary
  stochastic block model,~\cite{Celisse_etal} establish identifiability under the  assumption that the product vector $\alpha \pi$ (or
  $\pi^\intercal \alpha$) has distinct coordinates. This  condition is slightly stronger than requiring any two rows of
  the parameter matrix to be distinct. Another identifiability result appears in~\cite{bickel2011}
  for some specific block  models. These last two approaches are  dedicated to  the binary setup and cannot be generalized to the continuous case. 

Proposition   \ref{prop:ident}   does   not  cover   the undirected  affiliation
case,   where   only    two   intensities   $\alpha^{\text{in}}$   and
$\alpha^{\text{out}}$   are   considered   such   that    $\alpha^{(q,q)}=\alpha^\text{in}$                                   and
$\alpha^{(q,l)}=\alpha^\text{out}\ (q,l=1,\dots, Q;\ q\neq l)$.

\begin{proposition}
\label{prop:ident_affil}
  If  the     intensities    $\alpha^\text{in}$     and
  $\alpha^\text{out}$ are  distinct functions  on $[0,T]$, then both
  $\alpha^\text{in}$  and  $\alpha^\text{out}$   are  identifiable  on
  $[0,T]$ from the undirected affiliation Poisson process stochastic block model distribution $\Pt$ when  $n\ge 3$.  
  Moreover, for   $n\ge
  \max\{Q,  3\}$,  the  proportions   $\pi_1,\dots,  \pi_Q$  are  also
  identifiable up to a permutation.
\end{proposition}

\subsection{Additional notation}
\label{sec:proc_notation}
We introduce some quantities that are relevant in the following. Denote
\begin{align}
  Y^{(q,l)}_{\mathcal Z} &= \sum_{(i,j)\in\mathcal R}Z^{i,q}Z^{j,l} , \quad (q, l=1,\dots, Q), \label{Yql}\\
N_{\Z}^{(q,l)} &=\sum_{(i,j)\in\mathcal R} Z^{i,q}Z^{j,l} N_{i,j}, \quad (q, l=1,\dots, Q), \label{eq:N_Z}\\
Z_{m}^{(q,l)} &= Z^{i_m,q}Z^{j_m,l}, \quad (q, l=1,\dots, Q;\ m=1,\dots,M). \label{eq:Zmql}
\end{align}
These are the (latent) number of dyads $(i,j)\in \mathcal{R}$ with
latent   groups  $(q,l)$,   the  (latent)   counting  process   of
interactions   between  individuals   in   groups   $(q,l)$  and   the
(latent) binary indicator of observation $(t_m,i_m,j_m)$ belonging
to groups $(q,l)$, respectively. 
As these  quantities are unobserved,  our work relies on  proxies based on approximations of the latent group variables $Z^{i,q}$. Denote 
\begin{equation}
  \label{eq:tau_class}
\mathcal T=\left\{\tau=(\tau^{i,q})_{i=1,\dots,n,q=1,\dots,Q} : 
 \tau^{i,q}\ge 0, \sum_{q=1}^Q\tau^{i,q}=1, \quad (i=1,\dots,n;\ q=1,\dots,Q) \right\}.
\end{equation}
While the   variables $Z^{i,q}$ are indicators, their counterparts
$\tau^{i,q}$ are weights  representing the probability that node $i$ belongs to group $q$. 
Now,  for every  $\tau\in\mathcal T$,  replacing all  latent variables
$Z^{i,q}$ in \eqref{Yql}--\eqref{eq:Zmql} by $\tau^{i,q}$, we define 
$Y^{(q,l)}$, $N^{(q,l)}$ and $\tau_{m}^{(q,l)}$ which are estimators of
$Y^{(q,l)}_{\mathcal Z}$, $N_{\Z}^{(q,l)}$ and $Z_{m}^{(q,l)}$, respectively.

\section{Semiparametric estimation procedure}
\label{sec:method}
\subsection{A variational semiparametric expectation-maximization algorithm}
\label{sec:SPVEM}
The complete-data likelihood of observations $\mathcal O$ and latent variables $\mathcal Z=(Z_1,\dots,Z_n)$ is
\[
\mathcal L (\mathcal O,\mathcal Z \mid \theta)
 =\exp\left\{-\sum_{(i,j)\in\mathcal
    R}     A^{(Z_i,Z_j)}(T)
    \right\}\prod_{m=1}^M\alpha^{(Z_{i_m},Z_{j_m})}(t_m) \prod_{i=1}^n\prod_{q=1}^Q\pi_q^{Z^{i,q}}.
\]
The likelihood of the observed data $\mathcal L(\mathcal O\mid\theta)$ is
obtained by summing the above over the set of all
possible configurations of the latent variables $\mathcal  Z$. This set is huge and thus the likelihood  $\mathcal L(\mathcal O\mid\theta)$ is intractable  for direct maximization.  Hence, an
expectation-maximization algorithm~\citep{DempsterLR} is
used, which is an iterative  procedure especially adapted to cope with
latent variables.  It consists of an \texttt{E}-step and 
an \texttt{M}-step   that are iterated until convergence. 
In our model  two   issues arise. First, as already observed for the standard stochastic block model~\citep{Daudin_etal08}, the \texttt{E}-step    
requires the computation of the conditional distribution of
$\mathcal Z$ given the observations $\mathcal O$, which is not tractable.
Therefore, we use  a variational approximation~\citep{Jordan_etal} of the latent variables' conditional distribution to perform  the \texttt{E}-step. 
We refer to~\cite{MR_review} for a presentation and a discussion on the variational approximation in stochastic block models.  
Second,  part of  our parameter  is infinite  dimensional so  that the
\texttt{M}-step is partly replaced by a nonparametric estimation procedure, giving rise to a semiparametric
algorithm.

\subsection{Variational \texttt{E}-step}
The standard \texttt{E}-step consists in computing the expectation of the complete-data log-likelihood given the observations at some current parameter value $\theta$. 
This  requires  the  knowledge  of the  conditional  latent  variables
distribution  $\P_\theta(\mathcal  Z\mid\mathcal  O)$,  which  is  not
tractable because the $Z_i$'s are not conditionally independent.
We thus  perform  a   variational  approximation   of  $\P_\theta(\mathcal
Z\mid\mathcal  O)$  by a  simpler  distribution.  Using the  class  of
parameters $\mathcal{T}$ defined by~\eqref{eq:tau_class} we consider 
the family of factorized distributions on $\mathcal Z$ given $\mathcal O$ defined as
\begin{align*}
\P_\tau\{\mathcal Z=(q_1,\dots,q_n)\mid\mathcal O\}
&=\prod_{i=1}^n\P_\tau(Z_i=q_i\mid\mathcal O)
=\prod_{i=1}^n \tau^{i,q_i} ,\quad (q_1,\dots,q_n) \in\{1,\dots,Q\}^n,
\end{align*}
with corresponding expectation $E_\tau(\cdot\mid\mathcal O)$. Denoting
$\KL{\cdot}{\cdot}$ the Kullback-Leibler divergence,  we search for
the  parameter   $\hat  \tau\in\mathcal   T$  that  yields   the  best
approximation         $\P_\tau(\cdot\mid\mathcal        O)$         of
$\P_\theta(\cdot\mid\mathcal O)$ through
\begin{align}\label{eq:hat_tau}
\hat\tau = \argmin_{\tau\in\mathcal T}\KL { \P_\tau(\cdot\mid
  \mathcal{O}) } { \P_{\theta}(\cdot\mid \mathcal{O}) }.
\end{align}
Let $\mathcal{H}(\cdot)$ be the entropy of a 
distribution. It can be seen  that  $\hat\tau $ maximizes with respect to
$\tau$ the quantity 
\begin{align*}
J(\theta,\tau)  & =  E_\tau\left\{ \log\mathcal  L(\mathcal O,\mathcal
                 Z\mid\theta)\mid\mathcal                    O\right\}
                  +\mathcal{H}\{\P_\tau(\cdot\mid
  \mathcal{O})\} \\
&= -\sum_{q=1}^Q\sum_{l=1}^Q Y^{(q,l)} A^{(q,l)}(T) +
\sum_{q=1}^Q\sum_{l=1}^Q                  \sum_{m=1}^M\tau_{m}^{(q,l)}
  \log\left\{ \alpha^{(q,l)}(t_m)\right\} +
\sum_{i=1}^n\sum_{q=1}^Q{\tau^{i,q}}\log
  \left(\frac{\pi_q}{\tau^{i,q}} \right).
\end{align*}
The solution $\hat \tau$ satisfies a fixed point equation, which in practice
 is found by  successively updating the
 variational parameters $\tau^{i,q}$ via the following Equation~\eqref{eq_vstep_pi_q_fixed_point}  until convergence. 

\begin{proposition}
\label{prop:E_step}
The  solution  $\hat  \tau$  to  the  minimization  problem~\eqref{eq:hat_tau} satisfies the fixed point equation 
\begin{equation}
\label{eq_vstep_pi_q_fixed_point}
\hat \tau^{i,q}\propto \pi_q \exp\{D_{iq}(\hat \tau, \alpha)\}, \quad
(i=1,\dots,n;\ q=1,\dots,Q), 
\end{equation}
where  $\propto$  means  `proportional  to',  $\1_A$  denotes the
indicator function of set $A$ and 
\begin{align*}
  D_{iq}(\tau,\alpha)& = -\sum_{l=1}^Q \sum_{\substack{j=1\\
  j\neq i}}^n \tau^{j,l}
\left\{A^{(q,l)}(T)+A^{(l,q)}(T)\right\}
  \\
&\quad+\sum_{l=1}^Q\sum_{m=1}^M    \left[   \1_{\{i_m=i\}}\tau^{j_m,l}
  \log\left\{\alpha^{(q,l)}(t_m)\right\}+\1_{\{j_m=i\}}\tau^{i_m,l}\log\left\{\alpha^{(l,q)}(t_m)\right\}
  \right] . 
\end{align*}
\end{proposition}

\subsection{Nonparametric \texttt{M}-step: general description} 
 In  a parametric context the \texttt{M}-step consists in the maximization of $J(\theta,\tau)$
with respect to $\theta=(\pi,\alpha)$.  
Considering only  the finite-dimensional part $\pi$ of the parameter, the maximizer $\hat
\pi$ is 
\begin{equation}\label{eq_mstep_pi_q}
\hat \pi_q=\frac{\sum_{i=1}^n\tau^{i,q}}{\sum_{q=1}^Q\sum_{i=1}^n\tau^{i,q}}
=\frac1n{\sum_{i=1}^n\tau^{i,q}}, \quad  (q=1,\dots,Q).
\end{equation}
Concerning  the infinite-dimensional parameter $\alpha$,  we replace the
maximization of  $J(\pi,\alpha,\tau)$ with respect to  $\alpha$ by a
nonparametric estimation step. 
If the process $N_{\Z}^{(q,l)}$ defined  by~\eqref{eq:N_Z}  was observed, the estimation of $\alpha^{(q,l)}$ would be straightforward.
Now the  criterion $J$ depends on $\alpha$ through the quantity
$E_\tau  (\log  \mathcal  L(\mathcal   O,\mathcal  Z  \mid  \theta)  \mid
\mathcal O)$, that corresponds to the log-likelihood in 
  a setup where we observe the weighted cumulative process 
$N^{(q,l)}$ (see Section~\ref{sec:proc_notation}).  Intensities may be easily estimated in this direct observation setup.
We develop two different approaches for updating $\alpha$: a histogram and a kernel method.

\subsection{Histogram-based \texttt{M}-step}
\label{sec:histo}

In this part  the intensities $\alpha^{(q,l)}$ are estimated by  piecewise constant functions and we propose  a
data-driven choice of the partition of the time interval $[0,T]$.  The
procedure is based on a least-squares penalized criterion following the work
of~\cite{patricia}. Here $(q,l)$  is fixed  and for reasons of computational efficiency we    focus on regular  dyadic partitions  that form
embedded sets of partitions. 

For some given $d_{\max}$, consider all regular dyadic partitions of $[0,T]$ into $2^d$ intervals 
\[
\mathcal{E}_d = \left\{ E_{k} =\Big[(k-1)\frac  T {2^d} ; k \frac T{2^d}
  \Big) : k =1,\dots, 2^d\right\}, \quad (d= 0,\dots,d_{\max}).
\]
For any interval $E$ included in $[0,T]$, the estimated mean number of observed interactions $(t_m,i_m,j_m)$ between individuals with latent groups $(q,l)$ and occurring in  $E$ is 
\begin{align*}
N^{(q,l)}(E)=              \int_E             dN^{(q,l)}(s)              =
  \sum_{m=1}^M\tau^{i_m,q}\tau^{j_m,l}\1_{E}(t_m) .
\end{align*}
The total number of dyads $r$ is an upper bound for $Y^{(q,l)}$, so we
define a least-squares contrast  (relatively to the
counting process $N^{(q,l)}$) for all $f \in
\mathbb{L}^2([0,T],dt)$ by
\begin{equation*}
\gamma^{(q,l)}(f)=-\frac{2}{r} \int_0^T f(t) dN^{(q,l)}(t) +\frac{Y^{(q,l)}}{r} \int_0^T f^2(t) dt .
\end{equation*}
The  projection estimator  of
$\alpha^{(q,l)}$  on the  set  $S_d$ of  piecewise constant  functions
on $\mathcal{E}_d$ is 
\begin{equation*}
\hat{\alpha}_d ^{(q,l)}=\argmin_{f \in S_d} \gamma^{(q,l)}(f)= 
\frac1{Y^{(q,l)}}
\sum_{E\in \mathcal{E}_d}\frac{N^{(q,l)}(E)}{| E| }\1_{E}(\cdot),
\end{equation*}
where $|E|$ is the length of interval $E$, here equal to $T2^{-d}$. 
Now, adaptive estimation consists in choosing the best estimator among the
collection $\{ \hat{\alpha}_d^{(q,l)} : d=0,\dots, d_{\max} \}$.  
So  we  introduce an  estimator  of  the  partition through  $\hat d^{(q,l)} $
minimizing a penalized least-squares criterion of the form 
\begin{equation*}
crit^{(q,l)}(d)=\gamma^{(q,l)}(\hat{\alpha}_d^{(q,l)})+pen^{(q,l)}(d),
\end{equation*}
for
some penalty function $pen^{(q,l)} : \{0,\dots ,d_{\max}\}\rightarrow
\mathbb{R}^{+}$ that penalizes large partitions. 
We choose the penalty function as 
\begin{equation*}
pen^{(q,l)}(d)=\frac{2^{d+1}}{r} C \qquad\text{
with }\quad
C=\frac{ 2^{d_{\max}}}{T  Y^{(q,l)}} \sup_{E\in \mathcal{E}_{d_{\max}}}
N^{(q,l)} (E) .
\end{equation*}
Finally, the penalized least-squares criterion simplifies  to
\begin{equation*}
\hat d^{(q,l)} = \argmin_{ d  =0,\dots, d_{\max}} 2^d\Big\{-
\sum_{E\in \mathcal{E}_d}N^{(q,l)}(E)^2
   +  2^{d_{\max}  +1}  
\sup_{E'\in \mathcal{E}_{d_{\max}}} N^{(q,l)} \left(E'\right) \Big \}. 
\end{equation*}
The selected partition size $\hat d=\hat d^{(q,l)} $ may be specific to the  groups $(q,l)$. 
 Finally,  the adaptive estimator  of  intensity
$\alpha^{(q,l)}$ is 
\begin{equation}
\hat {\alpha}^{(q,l)}_{\text{hist}} (t) = \hat \alpha_{\hat d}^{(q,l)}(t)
  =\frac{2^{\hat d }}{T Y^{(q,l)}}
  \sum_{E\in \mathcal{E}_{\hat d}}N^{(q,l)} (E) \1_{E }(t) ,
\label{eq:mstep_histo}
\end{equation}
where $0\le t\le T ; q,l=1,\dots,Q$.
\cite{patricia}  develops her  approach  in  the Aalen  multiplicative
intensity   model,    which   is    slightly   different    from   our
context.  Moreover, our  setup  does not  satisfy  the assumptions  of
Theorem 1 in that reference as the number of jumps of the processes $N_{i,j}$
is not bounded by a known value (we have recurrent events). 
Nevertheless, in our simulations this procedure correctly estimates  the intensities $\alpha^{(q,l)}$ (see Section~\ref{sec:simus}).
We refer to~\cite{Birge_Baraud} for a theoretical study of an adaptive nonparametric estimation of the intensity of a
Poisson process. 
\cite{patricia}  also studies other  penalized least-squares estimators (for e.g. using  Fourier
bases), which might be used here similarly.

\subsection{Kernel-based \texttt{M}-step}
Kernel methods are suited to estimate smooth functions.  
If the variational parameters $\tau^{i,q}$ are good approximations of  the latent variables $Z^{i,q}$, then the intensity of  process $N^{(q,l)}$ (Section~\ref{sec:proc_notation}) is approximately
$Y^{(q,l)}\alpha^{(q,l)}$, where $Y^{(q,l)}$ is the variational mean number of
dyads with latent groups $(q,l)$. 
Following~\cite{Ramlau} and  considering  a nonnegative kernel function $K$ with
support within $[-1,1]$ together with some  bandwidth $b>0$,  the intensity $\alpha^{(q,l)}$ is  estimated by  
\begin{align}\label{eq_mstep__kernel_lamb_ql}
\hat{\alpha}_{\text{ker}}^{(q,l)}(t)
&=\frac{1}{b Y^{(q,l)} } \int_0^T K\Big(\frac{t-u}{b}\Big)dN^{(q,l)}(u) 
=\frac{1}{b Y^{(q,l)}} \sum_{m=1}^M\tau_{m}^{(q,l)} K\Big(\frac{t-t_m}{b} \Big), 
\end{align} 
if $Y^{(q,l)}>0$ and $\hat{\alpha}_{\text{ker}}^{(q,l)}(t)=0$ otherwise,
where $\tau_{m}^{(q,l)}$ is defined in Section~\ref{sec:proc_notation}.  

The bandwidth $b$ could be  chosen adaptively from the data following the procedure proposed
  by~\cite{gregoire}.   Kernel methods are not always suited to infer a function on a
  bounded interval as boundary effects may deteriorate their quality. However,  it is out of the scope of this work to
  investigate refinements of this kind.

\subsection{Algorithm's full description}
\label{sec:vem}
The  implementation  of   the  algorithm raises two  issues:
convergence  and  initialization. As  our  algorithm  is an  iterative
procedure, one has  to test for convergence. A  stopping criterion can
be defined based on criterion $J(\theta,\tau)$.
Concerning initialization the algorithm is run several times with
different starting values, which are chosen by some
k-means method combined with perturbation techniques 
 (see the Supplementary Material for details).  
Algorithm~\ref{algo:vem} provides a full description of the procedure.

 \begin{algorithm}[!h]
 \caption{Semiparametric variational expectation-maximization  algorithm}
 \label{algo:vem}
 \vspace*{-12pt}
 \begin{tabbing}
  \enspace $s \leftarrow 0$  \\
  \enspace Initialize $\tau^{[0]}$ \\
  \enspace \lWhile{convergence is not attained}{ \\
 \qquad Update $\pi^{[s+1]}$ via Equation~\eqref{eq_mstep_pi_q} with $\tau=\tau^{[s]}$ \\
 \qquad Update $\alpha^{[s+1]}$ via either Equation~\eqref{eq:mstep_histo} (histogram method)  or~\eqref{eq_mstep__kernel_lamb_ql} (kernel method) with $\tau=\tau^{[s]}$ \\
  \qquad Update $\tau^{[s+1]}$ via the fixed point Equation~\eqref{eq_vstep_pi_q_fixed_point} using $(\pi,\alpha)=(\pi^{[s+1]},\alpha^{[s+1]})$\\
 \qquad Evaluate the stopping criterion\\
\qquad $s\leftarrow s+1$ }\\
 \enspace Output $(\tau^{[s]},\pi^{[s]},\alpha^{[s]})$
 \end{tabbing}
 \end{algorithm}

\subsection{Model selection with respect to $Q$}
\label{sec:ICL}
We propose an integrated classification likelihood criterion that performs data-driven model
selection for the number of groups $Q$.  Roughly, this criterion is based on the complete-data
variational log-likelihood penalized by the number of parameters. It has been introduced in the mixture context in~\cite{BCG00}
and adapted to the stochastic block model in~\cite{Daudin_etal08}. 
The issue  here is that  our model  contains a nonparametric  part, so
that the  parameter is infinite  dimensional. However in the  case of
histogram estimators, once the partition  is selected, there is only a
finite number of parameters to estimate. This number can be used to build our integrated classification likelihood criterion.

For any $Q$ let $\hat \theta(Q)$ be the estimated parameter and $\hat
\tau(Q)$ the corresponding group probabilities obtained by our algorithm run with $Q$ groups. 
The parameter $\hat
\theta(Q)$ has two components: the first one  $\hat \pi(Q)$ is a vector of dimension
$Q-1$, while the second  has dimension
$\sum_{q, l=1,\dots, Q} \exp( \hat d^{(q,l)} \log 2)$. 
In the adaptation of the integrated classification likelihood criterion to the stochastic block model these
 components are  treated differently: the first one, that concerns the $n$ individuals is penalized by a $\log(n)/2$ term, while the second one
concerning the dyads is penalized by a  $\log(r)/2$ term. We refer to~\cite{Daudin_etal08} for more details. 
In our case the integrated classification likelihood criterion is
\begin{equation*}
\ICL (Q) = \log \mathbb{P}_{\hat \theta(Q)}\{ \mathcal{O},\hat
\tau(Q)\} -\frac 1 2 (Q-1)\log n - \frac 1 2 \log r 
\sum_{q=1}^Q\sum_{l=1}^Q2^{ \hat d^{(q,l)}}.
\end{equation*}
After fixing an upper bound $Q_{\max}$ we  select the number of groups 
\begin{equation}
  \label{eq:ICLmax}
\hat Q = \argmax_{ Q=1,\dots, Q_{\max}} \ICL(Q).
\end{equation}

\section{Synthetic experiments}
\label{sec:simus}

We generate data using the undirected Poisson process stochastic block
model in the following two scenarios.
\begin{enumerate}
\item  To evaluate  the  classification performance,  we consider  the
  affiliation model with $Q=2$ groups, equal probabilities $\pi_q=1/2$ and a  number of individuals $n$  in $\{10,30\}$. 
The intensities are sinusoids 
 with     varying    shifting     parameter    $\varphi$     set    to
 $\alpha^{\text{in}}(\cdot)=10 \{1+\sin(2\pi \cdot )\}$ and $\alpha^{\text{out}}(\cdot )=
  10 [1+\sin\{2\pi (\cdot +\varphi) ]$ with $\varphi \in \{0.01, 0.05, 0.1, 0.2, 0.5\}$. 
Clustering is more difficult for small values of $\varphi$. 

\item To  evaluate the  intensity estimators,  we choose  $Q=3$ groups
  with  equal probabilities  $\pi_q=1/3$ and  six intensity  functions
  that have different shapes and amplitudes (see continuous curves in Figure~\ref{estimated_intensities_Q3_n50}). 
The number of individuals $n$ varies in  $\{20,50\}$.   
\end{enumerate}
 
For   every parameter setting,  our algorithm is applied to 
 $1000$ simulated datasets. The  histogram estimator
is used with  regular dyadic partitions and $d_{\max}=3$, while the kernel estimator uses the Epanechnikov kernel.

\begin{figure}[t]
\centering
\includegraphics[width=\textwidth]{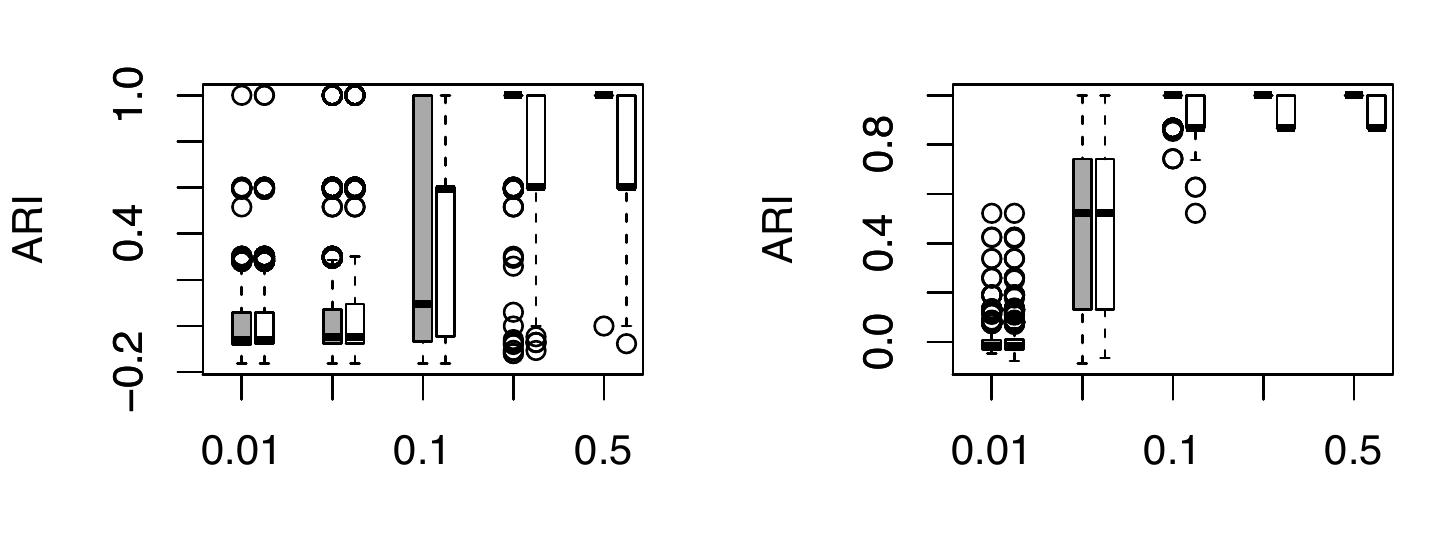}
  \caption{Boxplots   of  the   adjusted  rand   index  in   synthetic
    experiments from Scenario 1 for histogram (gray)
    and  kernel  (white)  estimators. The $x$-axis shows different  values  of
    $\varphi$ in $ \{0.01, 0.05, 0.1, 0.2, 0.5\}$.  Left panel: $n= 10$, right panel: $n=30$.}
  \label{boxplot_affiliation}
\end{figure}

To assess the clustering performance, we use the adjusted rand index~\citep{HA1985} that evaluates the
agreement between the estimated and the true latent structure. 
For two classifications that are  identical (up to label switching), this index equals $1$, otherwise the adjusted rand index is smaller than $1$ and  negative values are possible.
Figure~\ref{boxplot_affiliation}  shows the boxplots of the adjusted rand index obtained
with the histogram and the kernel versions of our method in Scenario 1. 
For  small values of the shifting parameter ($\varphi \in  \{0.01,0.05\}$),   the
intensities are so close that   classification is very difficult,
especially when $n=10$ is small. The classification improves  when  
the shift between the   intensities and/or the number of observations increase, achieving  (almost) perfect
classification for larger values of $\varphi$ and/or $n$.

Concerning the recovery of the   intensities   in Scenario 2,  we define the  risk  measuring the distance between the
true intensity  $\alpha^{(q,l)}$ and its estimate $\hat{\alpha}^{(q,l)}$ as 
\begin{equation*}
\textsc{Risk}(q,l)=\| \hat{\alpha}^{(q,l)} - \alpha^{(q,l)} \|_2= \Big[\int_0^T \{
    \hat{\alpha}^{(q,l)}(t) - \alpha^{(q,l)}(t) \}^2 dt \Big]^{1/2}.
\end{equation*}
Table~\ref{risk_Q3_n50} reports mean values and standard deviations of the risk when $n=50$ (The results
for $n=20$ are given in the Supplementary Material). 
Our histogram and kernel estimators are compared to their `oracle' equivalents obtained using the knowledge of the true
groups.  The table also reports the mean  number  of events with latent groups  $(q,l)$. 
As expected, when the true intensity is piecewise-constant,  the histogram version of
our method   outperforms the
kernel estimator. Conversely,  when the true intensity is  smooth, the kernel 
estimator is more appropriate to recover the shape of the intensity.  
Both estimators exhibit good performances with respect to the oracles versions. 

\begin{table}
\centering
\def~{\hphantom{0}}
\caption{Mean number of events and  risks with standard deviations (sd) in Scenario 2 with $n=50$.
Histogram (Hist) and kernel (Ker) estimators are compared with their oracle counterparts (Or.Hist, Or.Ker). 
All values associated with the risks are multiplied by 100.
}
{
\begin{tabular}{lccccc}
 Groups $(q,l)$ & Nb.events & Hist (sd)  & Or.Hist (sd)  & Ker (sd) & Or.Ker (sd) \\ 
\\
 $(1,1)$& 543 &  31 (32)& 20 (18) &  81 (51) & 63 (12)  \\ 
  $(1,2)$ & 949 & 44 (17)  & 81 (4) & 172 (57) & 156 (7)  \\ 
  $(1,3)$ & 545 & 46 (16) & 38 (6) & 53 (88) & 21 (6)  \\ 
  $(2,2)$ & 212 & 69 (10) & 70 (9) & 50 (56) & 35 (11)\\ 
  $(2,3)$ & 844 & 187 (6)  & 185  (2) & 125 (56) & 106 (11) \\ 
  $(3,3)$ & 298 & 83 (13) & 81 (13) & 64 (53) & 43 (12) 
\end{tabular}
}
\label{risk_Q3_n50} 
\end{table}

\begin{figure}[t]
\centering
 \includegraphics{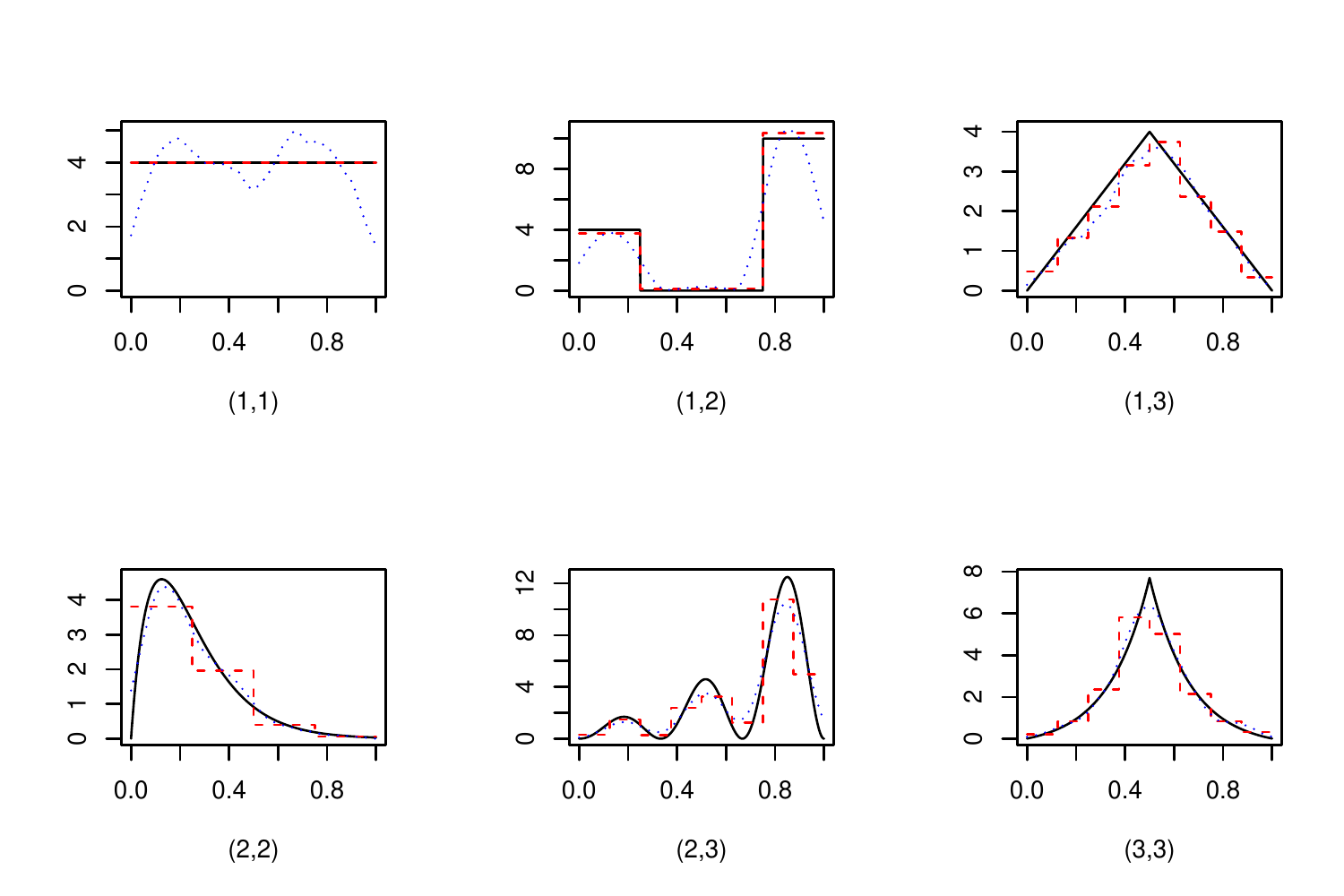}
\caption{Intensities in  synthetic experiments from Scenario  2 with
    $n=  50$.  Each panel  displays for a  pair  of groups  $(q,l)$  (given in the $x$-label) with
    $q,l=1,\dots, 3; q\le l  $  the  true intensity
    (continuous),  histogram  (dashed)  and  kernel  (dotted) estimates for one simulated dataset picked at random.}
\label{estimated_intensities_Q3_n50}  
\end{figure}

Finally, we use Scenario 2 to illustrate the performance of the integrated classification likelihood criterion to select the number $Q$ of latent groups from the data. For each  of the 1000 simulated datasets,
the  maximizer $\hat  Q$ of  the integrated  classification likelihood
criterion defined in~\eqref{eq:ICLmax} with $Q_{\max}=10$ is computed.
For $n=20$ the correct number of groups is
recovered in $74\%$ of the cases (remaining cases select values in $\{2,4\}$).  Moreover, for datasets where the criterion does not select the 
correct number $Q$, the adjusted rand index of the classification
obtained with 3 groups is rather low 
indicating that the algorithm has failed in the classification task and probably only a local maximum of the criterion $J$ has been found.

For $n=50$ our procedure selects the correct number of groups in $99.9\%$ of the cases.

\section{Datasets}
\label{sec:real_data}

\subsection{London cycles dataset}
We use the cycle hire usage data from the bike sharing system
  of the city of London from 2012 to 2015~\citep{bikes}. 
We  focus  on two randomly chosen weekdays, 1st (day1) and 2nd (day2) February 2012. 
Data consist in pairs of stations associated with a single
hiring/journey (departure station, ending station) and corresponding time stamp (hire time with second precision). 
The datasets have been pre-processed to remove journeys that either  correspond to loops,  last less than 1 minute or more than 3 hours or do not have an  ending station (lost or stolen bikes). 
The datasets contain $n_1=415$ and $n_2=417$ stations on day 1 and day 2 with $M_1=17, 631$ and $M_2= 16, 333$ hire
events respectively. With more than $\num{17e4}$ oriented pairs of stations the number of processes $N_{i,j}$ is huge, but only a
very small fraction -- around 7\% -- of these processes are non null (i.e. contain at least one hiring event between
these stations). 
Indeed bike sharing systems are mostly used for short trips and stations far from another are unlikely to be connected. 
Data correspond to origin/destination flows and are analysed in a  directed setup with the histogram version  of our
algorithm on a regular dyadic partitions with maximum size 32 ($d_{\max} =5$).  

The integrated classification likelihood criterion achieves its maximum with $\hat Q=6$ latent groups for both
datasets. Geographic locations of the bike stations and the clusters are represented on a city map (thanks to the
OpenStreetMap project), see Figure~\ref{fig:cycles_day1} for day 1. Clusters for day 2 are similar thus we focus on day
1. Our procedure globally recovers geographic clusters, as interacting stations
are expected to be  geographically close. A closer look at the clusters then
reveals more information.
While all but one clusters contain between 17 and 202 stations, cluster number 4 (dark blue crosses in
Figure~\ref{fig:cycles_day1}) consists of  only two bike stations.
They are located at Kings Cross railway  and  Waterloo railway stations and  are among the stations  with the highest
activities (for both departures and arrivals) in comparison to all other stations.
These stations appear to be `outgoing' stations in the morning with much more departures than arrivals around 8a.m. and
`incoming' stations at the   end of the day, with  more arrivals than  departures between 5p.m and   7p.m.  The bike
stations close to the two other main railway stations in London (Victoria and  Liverpool Street stations) do not exhibit
the same  pattern  and are clustered differently despite having a large number of hiring events. 
This cluster is thus created from the similarity of the temporal profiles of these two stations rather than their total
amount of interactions.   
Moreover it mostly interacts with cluster number 5 shown in light blue diamond in Figure~\ref{fig:cycles_day1} which roughly
  corresponds to the City of London neighbourhood.  It  is thus  characterized 
by stations used by people living in the suburbs and working in the city center. 

\begin{figure}[h]
  \centering
 \includegraphics[width=\textwidth]{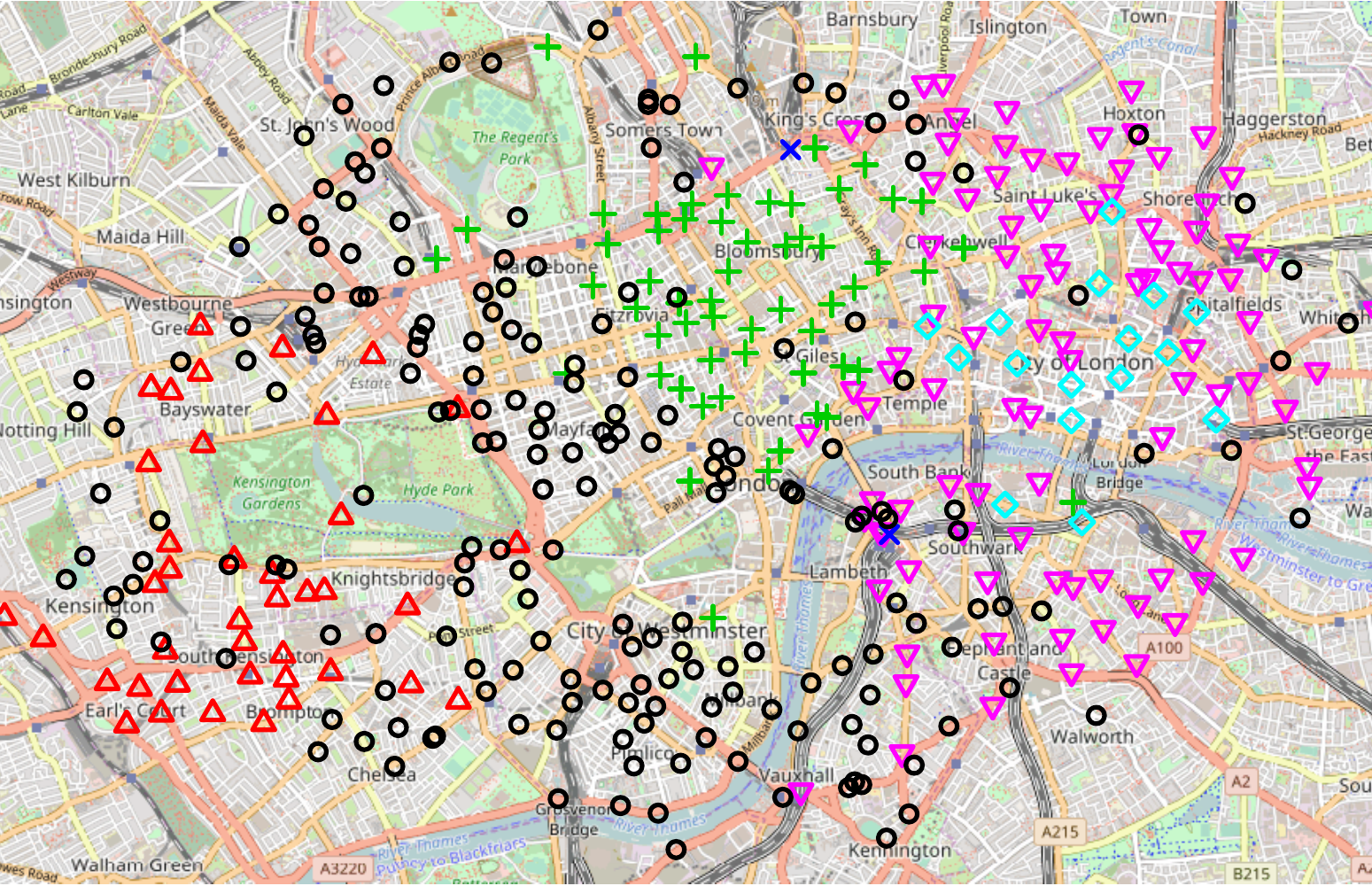}
  \caption{London bike sharing system: Geographic positions of the stations   
  and  clustering  into  six clusters (represented by  different  colors and symbols) for day 1. }
  \label{fig:cycles_day1}
\end{figure}

We compare our method with the discrete-time approach developed in \cite{Matias_Miele} where individuals are allowed to change groups during time. We applied  the corresponding \texttt{dynsbm} R package 
on the London bikes dataset aggregated into $T=24$ snapshots of one hour long, but no similar result came out of this:
their model selection criterion chooses 4 groups. 
The clusters with $Q=4$ groups in fact drop to only 1 nonempty group
between midnight to 3am, 2 groups between 3am and 7am and 3 groups  between 10pm and midnight. Globally there is one
very large group (containing from 168 to all stations, with mean number 285), one medium size group (from 0 to 148
stations, mean value 93) and 2 small groups (from 0 to 62 and 64 stations, mean values 14 and 22). Our small peculiar
cluster is not detected by the \texttt{dynsbm} method.

To conclude this section we mention that the same dataset is analysed in~\cite{Guigoures}
 with a different perspective.  \cite{velib} use a completely different approach, relying on Poisson mixture
models on a similar dataset of origin/destination flows (bike sharing system in Paris). Their approach does not take into account the network structure of the
data (where   e.g. two flows from the same station are related). As a consequence, clusters are obtained on pairs of
stations from which interpretation is completely different and in a way less natural. 



 \subsection{Enron dataset}
\label{sec:Enron}

\begin{figure}[t]
  \centering
 \includegraphics[width=0.8\textwidth]{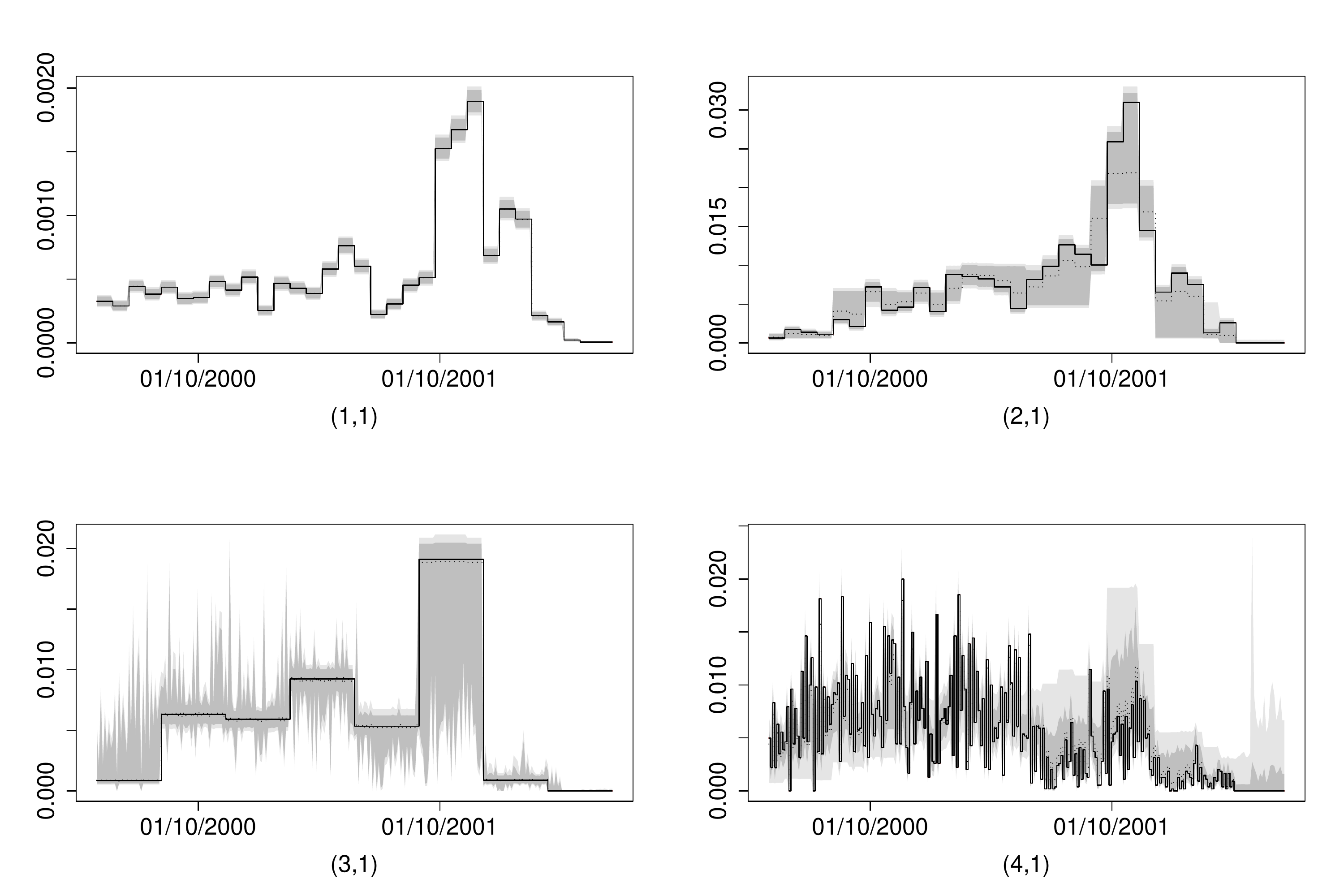}
  \caption{Enron: intensity estimates for   pair of groups $(q,1)$ for $q\in\{1,\dots,4\}$ and bootstrap confidence interval with confidence level 90\% (lightgrey) and 80\% (drakgrey).}
  \label{fig:enron_4intens}
\end{figure}
The Enron corpus contains emails exchanges
 among people working at Enron, mostly in the senior management, covering the period of  the affair that led to the bankruptcy of  the company~\citep{Enron}. 
 From the data provided by  the \cite{enron_data} we extracted  $21, 267$ emails exchanged among   147 persons between 27 April, 2000  and  14 June, 2002,  for  which  the sender,  the recipient and the time when the email was sent are known.  For most  persons their position in the company is known: one out of four are employees,  all other positions involve responsibility and we summarize  them as {\it managers}. 

We applied our algorithm using the directed model and the histogram approach with regular dyadic partitions of maximal  size $256$ ($d_{\max}=8$). For every fixed number of groups $Q\in\{2,\dots,20\}$ the algorithm  identifies one large group  always containing the same individuals. One may consider this group to have the standard behaviour at the company. For varying values of $Q$ the differences in the clustering only concern the remaining individuals with non standard behaviour. The integrated classification likelihood criterion does not provide a reasonably small enough number of groups that could be used for   interpretation. We thus choose to analyse the data with $Q=4$ clusters. 
The group with the standard behaviour  (group 1) contains 127 people both employees and managers. The second largest group (group 4) contains almost exclusively  managers, and group 2  is composed of employees.  
The standard behaviour (group 1) consists in little sending activity, but receiving mails from all other groups. The  members of group 3 have the most intense email exchange. The specific manager and employee groups (groups 2 and 4) are quite  similar 
and indeed at $Q=3$ both groups are merged together. 
For a number of pairs of groups communication is relatively constant over time, but for others the activity evolves a lot over time, as for instance for emails sent to members of group~1 (Figure~\ref{fig:enron_4intens}). 
The intra-group intensity of group 1 is increasing with a  peak in October 2001, which coincides with the beginning of the investigations related to the scandal. In contrast, the managers of group  4 communicate more intensely  with group 1 during the first half of the observation time than during the last half. More details are provided in the Supplementary Material.

The bootstrap confidence intervals in Figure \ref{fig:enron_4intens} are obtained via parametric bootstrap. In general, confidence intervals for large groups are very satisfactory, but this does not hold for    small groups. Indeed, here  some of the estimated probabilities $\hat \pi_k$ of group membership are very small (two are lower than 3\%), so  bootstrap samples tend to have empty groups. Evidently, this has a devastating effect on the  associated bootstrap intervals.

\cite{rastelli} analyse the same dataset with  a discrete-time model where individuals are  allowed to change
groups over time. They also obtain that most individuals are gathered in one group, which are rather inactive but receive
 some emails, and observe a specific behaviour of some groups of employees and managers. 
However, there is a difference in interpretation of results in the two models. In Rastelli's  model groups may be identified with specific tasks like sending newsletters for example and individuals may execute this task during a while and then change   activity. In contrast, our model identifies groups of individuals  with a similar behaviour over the entire observation period. Our approach thus is  natural  for the analysis of the temporal evolution of the activity of a fixed group of persons.

Finally,  we carried out a comparison with a classical stochastic block model. Indeed, taking $d_{\max}=0$ in our approach amounts to forget the timestamps of the emails, as the algorithm then only considers email counts over the whole observation period and our method boils down to  a classical stochastic block model with Poisson emission distribution and mean parameter $A^{(q,l)} (T)$~\cite[see for instance][]{Mariadassou_10}.
Comparing the classification at $Q=4$ to the one obtained with our continuous-time approach, the adjusted rand index
equals $0.798$ indicating that the clusterings are different. Indeed, both models identify roughly the same large group with standard behaviour, but  huge differences appear in the classification of the non standard individuals.
This illustrates that our approach does effectively take into account the temporal distribution of the events  to cluster the individuals. And the higher value of the complete log-likelihood in the Poisson process stochastic block model indicates that the solution is indeed an improvement to the classical model.


\section*{Acknowledgment}
{ We would  like to thank Agathe  Guilloux for pointing
  out valuable references, Nathalie Eisenbaum for her help on doubly stochastic counting processes and Pierre Latouche
  for sharing information on datasets.  The computations were partly performed at the institute for computing and data
  sciences at University Pierre and Marie Curie.}

\section*{Supplementary Material}
Supplementary  material available online at Biometrika   includes a complementary version of the model that specifically accounts
for sparse  datasets, the
proofs of  all theoretical results,  technical details on  the methods
and algorithm, some supporting results for the analysis of the two datasets explored here, as well as the study of a third dataset. 
A zip file containing  the computer code in R and the analysis of  the datasets is available at \url{http://cmatias.perso.math.cnrs.fr/Docs/ppsbm_files.tgz}



\newpage

\renewcommand{\thesection}{S.\arabic{section}}
\setcounter{section}{0}
\renewcommand{\theequation}{S.\arabic{equation}}
\setcounter{equation}{0}
\renewcommand{\thefigure}{S.\arabic{figure}}
\setcounter{figure}{0}

\begin{center}
{\Large Supplementary material for:  A semiparametric extension of the
  stochastic block model for longitudinal networks}\\

\end{center}

All  the references  are from  the main  manuscript, except  for
  those appearing as S-xx that are within this document.

\section{Identifiability proofs}
\begin{proof}[Proof of Proposition~\ref{prop:ident}]
The proof follows ideas similar to those of Theorem 12 in~\cite{AMR_JSPI}.   
For notational convenience this proof is presented in the undirected
setup where the set of intensities is $\alpha=\{\alpha^{(q,l)}: q,l=1,\dots, Q;\ q \leq l\}$. 
The directed case is treated in the same way.

To explain the general idea of the proof we  start by  considering  the distribution  of  one marginal  process
$N_{i,j}$. This is a Cox process directed by the random measure 
\[
A_{i,j} \sim \sum_{q=1}^Q\sum_{l=1}^Q \pi_q\pi_l \delta_{A^{(q,l)}}. 
\]
Here for  any $q\leq  l $,  we use the  notation $A^{(q,l)}$  for the
measure on $[0,T]$ defined by $ A^{(q,l)}(I) =\int_I \alpha^{(q,l)}(u)
du$ for all measurable $I\subset [0,T]$. We also  recall that $\delta_u$ is the Dirac mass at point $u$.
It is known that the mapping of probability laws of random measures into laws of Cox processes directed by them is a
bijection~\citep[see for example Proposition 6.2.II in][]{Daley_Vere}. In other words the distribution of $N_{i,j}$
uniquely determines the finite measure (on the set of measures on $[0,T]$) 
\[
 \sum_{q=1}^Q\sum_{l=1}^Q\pi_q\pi_l \delta_{A^{(q,l)}}.
\]
According to Assumption~1  the intensities $\alpha^{(q,l)}$
 are distinct. Hence, the corresponding measures
$A^{(q,l)}$ are all different and we may recover from the distribution of our counting process $N_{i,j}$ the set of
values 
$\{(\pi_q^2, A^{(q,q)}) : q=1,\dots, Q\}\cup\{(2\pi_q\pi_l, A^{(q,l)})
:q,l=1,\dots, Q;\ q<l\}$
or equivalently the set
$\{(\pi_q^2,\alpha^{(q,q)})   :   q=1,\dots, Q\}\cup\{(2\pi_q\pi_l,
\alpha^{(q,l)});\ q,l=1,\dots, Q;\ q<l\}$. In particular we recover the functions $\alpha^{(q,l)}$ almost everywhere on $[0, T ]$ up to a permutation 
on the pairs of groups $(q,l)$.
However for the recovery up to a permutation in $\mathfrak{S}_Q$ it is necessary to consider higher-order marginals.

We now fix three distinct integers $i,j,k$ in $\{1,\dots,n\}$ and consider the trivariate counting process $(N_{i,j},N_{i,k},N_{j,k})$. In the same
way, these are Cox processes directed by the triplet of random measures  $(A_{i,j},A_{i,k},A_{j,k})$ such that 
\begin{equation*}
(A_{i,j},A_{i,k},A_{j,k}) \sim \sum_{q=1}^Q\sum_{l=1}^Q \sum_{m=1}^Q\pi_q\pi_l\pi_m \delta_{(A^{(q,l)}, A^{(q,m)}, A^{(l,m)}) }. 
\end{equation*}
We write this distribution in such a way that distinct components appear only once 
\begin{align}
  \label{eq:trivariate}
\sum_{q=1}^Q & \pi_q^3 \delta_{(A^{(q,q)}, A^{(q,q)}, A^{(q,q)}) } \nonumber \\
+\sum_{q=1}^Q\sum_{\substack {l=1 \\ l\neq q}}^Q & \pi_q^2\pi_l \Big\{\delta_{(A^{(q,q)}, A^{(q,l)}, A^{(q,l)}) } + \delta_{(A^{(q,l)}, A^{(q,q)}, A^{(q,l)}) } +\delta_{(A^{(q,l)}, A^{(q,l)}, A^{(q,q)}) }\Big\} \nonumber\\
+ \sum_{q=1}^Q\sum_{\substack {l=1 \\ l\neq q}}^Q \sum_{\substack {m=1 \\ m\neq q,l}}^Q &\pi_q\pi_l\pi_m  \delta_{(A^{(q,l)},  A^{(q,m)}, A^{(l,m)}) }. 
\end{align}
Using  the  same  reasoning  we   identify  the  triplets  of  values
$\{(A^{(q,l)} , A^{(q,m)}, A^{(l,m)} ) : q,l,m =1,\dots, Q\}$ up to a permutation on the triplets $(q,l,m)$.  Among these, the only values
with three identical components are   $\{(A^{(q,q)}  ; A^{(q,q)} ; A^{(q,q)}  ) :  q=1,\dots, Q\}$
and thus the measures $\{ A^{(q,q)} : q=1,\dots, Q\}$ are identifiable 
up to a permutation in $\mathfrak{S}_Q$. Going back to~\eqref{eq:trivariate} and looking for the Dirac 
terms  at points  that  have  two identical  components  (of the  form
$( A^{(q,q)} , A^{(q,l)}, A^{(q,l)})$ and two other similar terms 
with permuted components),  we can now identify the set of measures 
\[
\{ (A^{(q,q)}, \{A^{(q,l)} : l=1,\dots, Q ;\ l \neq q \}) : q=1,\dots, Q \}.
\]
This  is   equivalent  to  saying   that  we  identify   the  measures
$\{A^{(q,l)} : q,l=1,\dots, Q ;\ q\leq l\}$ up to a permutation
in $\mathfrak{S}_Q$. Obviously this also identifies the corresponding
intensities $\{\alpha^{(q,l)} : q,l=1,\dots, Q;\ q\leq l \}$ almost everywhere on $[0,T]$ up to a permutation in $\mathfrak{S}_Q$.
To finish the proof we need to identify the proportions $\pi_q$. Note that as we identified the
components    $\{A^{(q,q)}    :    q=1,\dots,   Q\}$,    we    recover
from~\eqref{eq:trivariate} the set of values
$\{\pi_q^3 : q=1,\dots, Q\}$ 
up to the same permutation as on the $A^{(q,q)} $'s. This concludes the proof. 
\end{proof}


\begin{proof}[Proof of Proposition~\ref{prop:ident_affil}]
Note that the setup considered here is undirected. 
We follow some of the arguments already appearing in the proof of Proposition 1. Let $A^{\text{in}}$
(resp. $A^{\text{out}}$) denote the measure whose intensity is $\alpha^{\text{in}}$ (resp. $\alpha^{\text{out}}$). 
The univariate process $N_{i,j}$ is a Cox process directed by the random measure  $A_{i,j}$ that is now distributed as 
\[
A_{i,j}\sim (\sum_{q=1}^Q \pi_q^2) \delta_{A^{\text{in}}} + 
(\sum_{q=1}^Q\sum_{\substack {l=1 \\ l\neq q}}^Q \pi_q\pi_l)
\delta_{A^{\text{out}}} . 
\]
Thus the measures $A^{\text{in}}$ and $A^{\text{out}}$ are identifiable from the distribution of $N_{i,j}$, but only up to a
permutation. Similarly to the previous proof we rather consider the trivariate Cox processes $(N_{i,j},N_{i,k},N_{j,k})$ directed by the
random measures $(A_{i,j},A_{i,k},A_{j,k})$  whose distribution in the affiliation case has now five atoms
\begin{multline*}
 \Big(\sum_{q=1} ^Q \pi_q^3 \Big) \delta_{(A^{\text{in}} ,A^{\text{in}} , A^{\text{in}} )}    
 + \Big(\sum_{q=1}^Q\sum_{\substack {l=1 \\ l\neq q}}^Q \pi_q^2\pi_l \Big) \delta_{(A^{\text{in}} , A^{\text{out}}, A^{\text{out}})}   
 + \Big(\sum_{q=1}^Q\sum_{\substack {l=1 \\ l\neq q}}^Q \pi_q^2\pi_l \Big) \delta_{(A^{\text{out}}, A^{\text{in}}, A^{\text{out}} )} \\ 
 + \Big(\sum_{q=1}^Q\sum_{\substack {l=1 \\ l\neq q}}^Q  \pi_q^2\pi_l \Big) \delta_{(A^{\text{out}}, A^{\text{out}}, A^{\text{in}} )} 
+ \Big(\sum_{q=1}^Q\sum_{\substack {l=1 \\ l\neq q}}^Q \sum_{\substack {m=1 \\ m\neq q,l}}^Q  \pi_q\pi_l\pi_m\Big) \delta_{(A^{\text{out}} , A^{\text{out}} ,
  A^{\text{out}} )}  .
\end{multline*}
As previously, these five components are identifiable up to a permutation on $\mathfrak{S}_5$. Now it is easy to identify the three
components for which two marginals have same parameters and the third one has a different parameter. Thus we recover
exactly the measures  $A^{\text{in}}$ and $A^{\text{out}}$. 
This also identifies the corresponding intensities $\alpha^{\text{in}}$ and $\alpha^{\text{out}}$ almost everywhere on $[0,T]$. 

Now identification of the proportions $\{\pi_q : q=1,\dots, Q\}$ follows an argument already used in the proof of Theorem 13
in~\cite{AMR_JSPI} that we recall here for completeness.  From the trivariate distribution of
$(N_{i,j},N_{i,k},N_{j,k})$ and the already recovered values $A^{\text{in}}$ and $A^{\text{out}}$,  we
identify the proportion 
$\sum_{q=1}^Q \pi_q^3$.
Similarly for any $n\ge 1$, by considering the multivariate distribution of
$(N_{i,j})_{(i,j)\in \mathcal{R}}$, we can  identify the Dirac mass at
point  $(A^{\text{in}}, \dots,A^{\text{in}}  )$ and  thus its  weight 
which is equal to $\sum_{q=1}^Q \pi_q^n$.  
By the Newton identities the values $\{ \sum_{q=1}^Q \pi_q^n : 
n=1,\dots, Q\}$ determine the values of elementary symmetric
polynomials $\{\sigma_n(\pi_1,\dots,\pi_Q) : n=1,\dots, Q\}$. These, in turn, are (up to sign) the coefficients of the
monic polynomial whose roots (with multiplicities) are precisely $\{\pi_q :  q=1,\dots, Q\}$. Thus  the proportion 
parameters are recovered up to a permutation. 
\end{proof}


\section{Variational \texttt{E}-step: Proof of Proposition~\ref{prop:E_step}}
\label{sec:other_proof}
For the Kullback-Leibler divergence we compute
\begin{align*}
&  \KL  {\P_\tau(\cdot  \mid \mathcal{O})  }  {\P_{\theta}(\cdot  \mid
  \mathcal{O}) } 
=  E_\tau\left\{  \log\frac{   \P_\tau(\mathcal  Z  \mid  \mathcal{O})}{\P_{\theta}(\mathcal Z \mid \mathcal{O})}  \mid \mathcal{O}  \right\}
= E_\tau\left\{ \log\frac{ \P_\tau(\mathcal Z\mid \mathcal{O})
\P_{\theta}(  \mathcal{O})  }{\mathcal  L(\mathcal  O,\mathcal  Z\mid
\theta)} \mid \mathcal{O}\right\}\\
&\qquad=\sum_{i=1}^nE_\tau\left(\log\tau^{i,Z_i}     \mid     \mathcal
  O\right) + \log\P_{\theta}( \mathcal{O})-E_\tau\left\{ \log\mathcal
  L(\mathcal O,\mathcal Z \mid \theta) \mid \mathcal{O}\right\}.
\end{align*}
The      complete-data
log-likelihood $\log \mathcal L(\mathcal O,\mathcal Z\mid \theta)$ is
\begin{equation*}
-\sum_{q=1}^Q\sum_{l=1}^Q Y^{(q,l)}_{\mathcal Z} A^{(q,l)}(T)+
\sum_{q=1}^Q\sum_{l=1}^Q \sum_{m=1}^M Z_{m}^{(q,l)} \log\left\{ \alpha^{(q,l)}(t_m)\right\}+
\sum_{i=1}^n\sum_{q=1}^Q{Z^{i,q}}\log \pi_q , 
\end{equation*}
where $Y^{(q,l)}_{\mathcal Z}$ and $Z_{m}^{(q,l)}$ have been introduced in Equations~\eqref{Yql} and~\eqref{eq:Zmql}, respectively. 
Now note that $E_\tau(Z^{i,q}\mid \mathcal O) 
=\P_\tau(Z^{i,q}=1\mid \mathcal O)
= \P_\tau(Z_i=q\mid \mathcal
O)=\tau^{i,q}$. 
 Moreover by the factorised form of $\P_\tau$, for every
$i\neq j$ we have 
$$
E_\tau(Z^{i,q} Z^{j,l}\mid \mathcal O)=E_\tau(Z^{i,q}\mid \mathcal O )
E_\tau(Z^{j,l}\mid \mathcal O)=\tau^{i,q}\tau^{j,l} .
$$
The quantity $Y^{(q,l)}$ 
is thus equal to $E_\tau(Y^{(q,l)}_{\mathcal Z}\mid \mathcal{O})$, namely the
variational  approximation  of the  mean  number  of dyads  with latent groups
$(q,l)$. Similarly $\tau_{m}^{(q,l)}$ equals $E_{\tau}(Z_{m}^{(q,l)}\mid \mathcal O)$, 
 the variational approximation of the probability that observation $(t_m,i_m,j_m)$ corresponds to a dyad with latent groups 
$(q,l)$.
It follows that 
\begin{align*}
\hat\tau = \argmin_{\tau\in\mathcal T}\KL{\P_\tau(\cdot\mid
  \mathcal{O}) } { \P_{\theta}(\cdot\mid \mathcal{O}) } 
=\argmax_{\tau\in\mathcal T}J(\theta,\tau),
\end{align*}
where $J(\theta,\tau)$ is 
\begin{align}
\label{eq:Jtau}
-\sum_{q=1}^Q\sum_{l=1}^Q Y^{(q,l)} A^{(q,l)}(T) +
\sum_{q=1}^Q\sum_{l=1}^Q                  \sum_{m=1}^M\tau_{m}^{(q,l)}
  \log\left\{ \alpha^{(q,l)}(t_m)\right\} +
\sum_{i=1}^n\sum_{q=1}^Q{\tau^{i,q}}\log
  \left(\frac{\pi_q}{\tau^{i,q}} \right).
\end{align}
The variational \texttt{E}-step consists in maximizing $J$ with respect to the $\tau^{i,q}$'s which are constrained 
to satisfy $\sum_{q=1}^Q \tau^{i,q}=1$ for all $i$. In other words  we maximize 
\begin{equation*}
M(\tau,\gamma)=J(\theta, \tau)+\sum_{i=1}^n\gamma_i\left(\sum_{q=1}^Q\tau^{i,q}-1\right),
\end{equation*}
with Lagrange multipliers $\gamma_i$. 
The partial derivatives are
\begin{align*}
\frac\partial{\partial \tau^{i,q}}M(\tau,\gamma)
&=-\sum_{l=1}^Q \sum_{j\neq i} \tau^{j,l}\left\{ A^{(q,l)}(T) +A^{(l,q)}(T) \right\}+
\sum_{l=1}^Q\sum_{m=1}^M          \1_{\{i_m=i\}}          \tau^{j_m,l}
  \log\left\{ \alpha^{(q,l)}(t_m)\right\} \\
&\quad+\sum_{l=1}^Q\sum_{m=1}^M       \1_{\{j_m=i\}}      \tau^{i_m,l}
  \log\left\{ \alpha^{(l,q)}(t_m)\right\}
+ \log \left( \frac{\pi_q}{\tau^{i,q}} \right) -1+\gamma_i, \\ 
\frac\partial{\partial \gamma_{i}}M(\tau,\gamma)
&=\sum_{q=1}^Q\tau^{i,q}-1.
\end{align*}
The partial derivatives are null if and only if $\sum_{q=1}^Q\tau^{i,q}=1$ and the 
$\tau^{i,q}$'s satisfy the fixed point equations~\eqref{eq_vstep_pi_q_fixed_point}, with
$\exp(\gamma_i-1)$ being the normalizing constant.

\section{Details on the Algorithm}
\label{sec:init}
Initialisation is a crucial point for any clustering method. Our  variational expectation-maximization algorithm starts with a classification of the nodes and
iterates an \texttt{M}-step followed by a variational \texttt{E}-step. 
We apply the algorithm on  multiple initial classifications of the nodes based on
various aggregations of the data: on the whole time interval and on sub-intervals. Sub-intervals are obtained through 
regular dyadic partitions of $[0,T]$ (parameter \texttt{l.part} in the R code) 
and a $k$-means algorithm applied on  the
rows of the adjacency matrices  of each of  the aggregated datasets provides   starting points for our algorithm.
To obtain further starting values we   use perturbations of 
the different $k$-means classifications: a given percentage of the total number of individuals is picked at random (parameter \texttt{perc.perturb}) and their group
memberships are  shuffled. The perturbation procedure can be applied several  times (parameter \texttt{n.perturb}). 
The algorithm returns as final result the run that achieves the largest value of criterion $J$.

There is no theoretical result on the existence of a solution to the fixed point
equation~\eqref{eq_vstep_pi_q_fixed_point}. 
The iterations for the fixed point equation 	are initialized with the value of $\tau$ obtained at the
previous variational \texttt{E}-step. In practice,  convergence is fast and we stop the fixed-point iterations either when convergence is achieved ($|\tau^{[s]}-\tau^{[s-1]}|< \varepsilon = 10^{-6}$)  or when the maximal number of iterations is attained (\texttt{fix.iter}=$10$).

As the variational expectation-maximization algorithm aims at maximizing $J$ defined in \eqref{eq:Jtau}, the algorithm
is stopped when the increase of $J$ is less than a given threshold ($\varepsilon= 10^{-6}$), that is when
$$
\left|\frac{J(\theta^{[s+1]},\tau^{[s+1]})-J(\theta^{[s]},\tau^{[s]})}{J(\theta^{[s]},\tau^{[s]})}\right|<\varepsilon , 
$$
or when the maximal number of iterations has been attained (\texttt{nb.iter}=$50$).

\section{Additional tables and figures}
Figure~\ref{fig_intensities_Q2} shows the intensities used in Scenario
1 from the synthetic experiments to assess the clustering performances of our method. 

Table~\ref{risk_Q3_n20} gives the risks $\textsc{Risk}(q,l)$ and  standard
deviations of the  histogram and the kernel versions of  our method as
well as oracle quantities (obtained with known  group labels) in Scenario 2 when
$n=20$ (this is the analogue of Table~\ref{risk_Q3_n50} in the main manuscript where $n=50$).

Figure~\ref{boxplot_Q3}  shows boxplots  of  the  adjusted rand  index
obtained from  the synthetic experiments  from Scenario 2.  They are 
computed over 1000 datasets with  different  numbers $n$
of individuals and for the two estimation methods (histogram and kernel).

Figure~\ref{Qbest-20} shows the model  selection results on the number
of groups based on the integrated classification likelihood in Scenario 2
with $n=20$: the left panel  shows the frequency  of the  selected values
$\hat Q$  over the 1000  datasets; the right  panel shows boxplots  of the
adjusted rand index  between the estimated classification with 3  groups and the
true latent structure as a function of the number of groups  selected by the  integrated classification likelihood criterion. 
On this right panel, one can see that when the criterion does not select the 
correct number $Q$ the adjusted rand index of the classification
with three groups is rather low
indicating that the algorithm has failed in the classification task and probably only a local maximum of the criterion $J$ has been found.
 \\

Turning to the London bike sharing system dataset, Figure~\ref{fig:cycles_temp} shows the temporal profiles of the 2 stations in the smallest cluster for day 1. One can see that these are `outgoing' stations around 8am and  `incoming' stations between 5 and
7pm. 
Figure~\ref{fig:cycles_intens} shows the highest intensities estimated by our model between these 6
clusters, all other intensities are almost null. The most important interactions occur from cluster 4 (the smallest cluster) to cluster 5 (`City of London'
cluster) in the morning and conversely from cluster 5 to cluster 4 at the end of the day. \\

Concerning  the analysis of the Enron dataset  we  provide here  additional   tables and figures for $Q=4$ groups.  Table~\ref{tab_enron_groupcomposition} gives the size and  composition of the four groups. For a part of the persons in the dataset the position at Enron is not available. This is the reason why the total size of the group  sometimes exceeds the sum of the number of managers and employees in the group. 
 
 Figure~\ref{enron_fig_heatmap} gives the logarithm of the mean values of the estimated intensities $\alpha^{(q,l)}$. 
 The lack of symmetry of the matrix indicates that communication is far from being symmetric and that
 the use of the directed model is appropriate.
 
Finally, Figure~\ref{fig_enron_all_intensities} shows the estimated intensities and associated bootstrap confidence intervals. 
 The bootstrap intervals are obtained by parametric bootstrap. More precisely,  every bootstrap sample contains the same number of individuals (here $n=147$) and for every individual $i$ the group membership $Z_i^*$ is drawn from the multinomial distribution  $\mathcal M(1,\hat\pi)$ where $\hat\pi$ is the vector of  group probabilities estimated from the data.  Then for every pair of individuals  $(i,j)$ realizations from a Poisson process with intensity $\hat\alpha^{(Z_i^*,Z_j^*)}$ are simulated, where $\{\hat\alpha^{(q,l)}\}_{q,l=1,\dots,Q}$ denote the estimated intensities. 
 Finally,  bootstrap confidence intervals are obtained by the   percentile method.
 
Here the bootstrap intervals suffer from the fact that some of the group probabilities $\hat \pi_k$ are very low, implying that the probability that a bootstrap sample contains empty groups is relatively high ($0.15$). That is, about 15\%  of the bootstrap samples do not provide any information on some of the intensities $\alpha^{(q,l)}$, implying that the associated estimators are completely erroneous. The groups that are the most concerned by the problem are  group 2 and 3.


\section{Additional example: Primary school temporal network dataset}
To understand contacts between children at school and to quantify the
transmission opportunities of respiratory infections, data on
face-to-face interactions   in a French  primary school were collected. 
The dataset is presented in detail in~\cite{Stehle} and available
online~\citep{sociopatterns}.
Children are aged from 6 to 12 years and the school is composed of five grades, each of them comprising two classes, for a total
of 10 classes (denoted by $1A, 1B, \dots, 5A,5B$). Each class has an assigned
teacher and an assigned room. The school day runs from  8.30am to 4.30pm,
with a lunch break from 12pm to 2pm and two breaks of 20-25 min in the
morning and in the afternoon. 
Lunch is served  in a  common canteen  and a
shared playground is located outside the main building.  As the
playground and the canteen do not have enough capacity to host
all pupils at a time, only two or three classes have
breaks together, and lunch is served in two 
turns.
The dataset contains $125, 773$  face to face contacts among $n=242$
individuals ($232$ children and $10$ teachers) 
observed during  two days. We applied our procedure in the undirected setup with histograms based
on  a regular dyadic  partitions  with  maximum  size  256 ($d_{\max}  =8$).  

The integrated classification likelihood criterion achieves its maximum
with $\hat Q=17$ latent groups. 
Figure~\ref{primaryschool_tau} shows the clustering of the $n$
individuals into  the $17$ groups, where children from different
classes are represented with different colors.  
Some groups correspond exactly or almost exactly to classes 
(for example, group  9 consists almost perfectly of class 1A), whereas other classes are split into several groups (for example class 1B is splited into groups 1 and 16). Moreover, one group (group  6) corresponds to the entire class 4B with, in addition, pupils coming from almost all other classes. Teachers never form a
  particular group apart, but they are generally in the cluster of their assigned class.

The highest intensities are the
intra-group intensities. As groups mainly correspond to classes, this
highlights  that most  contacts  involve individuals of  the same  class
and that the dataset is structured into communities ({i.e.}  groups  of  highly  connected   individuals and with  few  inter-groups  interactions). Figure~\ref{primaryschool_intra_group} shows the estimated intra-group intensities for each group
with at least 3 individuals.
Peaks   of interactions are observed during the two breaks in the morning and in the afternoon.   At  lunch time
interactions between children vary from the first to the second day and are less important than during the breaks when
they play together. We also observe periods with no interaction at all. For example, the estimated intensity of group 9
(class 1A) is null between 3:30am and 4am suggesting that 
some particular school activity like sports takes place during which contacts were not observable for technical reasons.
The group number 6 composed with the entire class 4B and others pupils
clearly appears as the group with the lowest intra-group intensity. This
means that this cluster gathers the individuals having less interactions with 
others. Class 4B also appears as the class having the least intra-class interactions in~\cite{Stehle}.

Concerning   inter-group  connections most of the estimated intensities for groups  $(q,l)$ with  $q\neq l$ can be  considered as null, except  for some
that we discuss now. 
First, as  our procedure splits some  children of the same class into separate  groups,   the  inter-group  interactions associated with these clusters  correspond in
fact to intra-class interactions. For example,  class 1B is split into group 1 (with 18
 pupils) and group 16 (with 7 pupils). The estimated inter-group
 intensity shows that those two groups interact.  Our clustering has
 formed two separate groups because group 1 has more
intra-group interactions than the other (see Figure~\ref{primaryschool_intra_group}).

Second,  intensities  between groups made of children of the same grade
are significant,  suggesting that children mostly interact with children
of the same age (see e.g. Figure~\ref{primaryschool_samegrade} that shows the case $(q,l)=(5,13)$). Those interactions
are observed during the two breaks  in the morning and in the afternoon as well as at  lunch time.

Third, the estimated intensities suggest particular behaviour of some
pupils. For example, concerning class 2B
(except the two pupils assigned to group 6), which is separated into
group 12 (with 21 pupils), group 11 (with 2 pupils) and group 17 (with
only one pupil), the  estimated intensities (see
Figure~\ref{primaryschool_class_2B}) suggest first
that the two children in group 11 
have very strong interaction with the 
pupil in group 17 (notice the different $y$-scale used in the Figure),
and second that those interactions do not occur during the lunch time.

Similar results are obtained for classes 2A and 5B. This means that our procedure detects subgroups of pupils with a specific  behaviour.

\section{The sparse setup: theory}
We consider an extended  setup where some of the
processes $N_{i,j}$ may have a null intensity.  We thus introduce 
additional latent variables $U_{i,j} \in \{0,1\},\ ((i,j)\in \mathcal{R})$ that conditional
on the $Z_i$'s are independent Bernoulli with $\beta_{q,l}$ being the
parameter of the distribution of $U_{i,j}$ conditional on $Z^{i,q}Z^{j,l}=1$. 
We keep the global conditional independence
assumption  by imposing  that conditional  on $(Z_i,U_{i,j})_{(i,j)\in
  \mathcal{R}}$ the counting processes $(N_{i,j})_{(i,j)\in
  \mathcal{R}}$ are independent. 
  Then   conditional  on   $U_{i,j},
Z_i , Z_j$, the counting process $N_{i,j}$ is an inhomogeneous
Poisson process with intensity 
\[
U_{i,j}                                \alpha^{Z_i,Z_j}=\sum_{q,l=1}^Q
U_{ij}Z^{i,q}Z^{j,l}\alpha^{(q,l)}, \quad ( (i,j)\in \mathcal R).
\]
In this way the  additional  latent   variable $U_{i,j}$  accounts  for  sparsity   in  the
interaction processes; in each pair of groups $(q,l)$, there is now
a proportion $1-\beta_{q,l}$ of  dyads $(i,j)\in \mathcal R$ that do
not interact  (so that the  corresponding process $N_{i,j}$  is almost
surely  0). As  such these  non interacting  dyads will  not tend  to
decrease the estimate of the common intensity $\alpha^{(q,l)}$. Moreover clustering in this model should give different
groups, less driven by the absence of interactions. 

We let $\mathcal{U}=(U_{i,j})_{(i,j)\in \mathcal{R}}$ and the parameter value
is $\theta=(\pi, \beta,\alpha)$. \\

Identifiability may  be proved  under the same  assumptions, requiring
moreover that none of the intensities $\alpha^{(q,l)}$ is itself equal
to zero. We discuss this in the undirected case (similarly to the identifiability proof of the main model). 
 Indeed $N_{i,j}$ is now  a counting process directed  by the
random measure $A_{i,j}$ whose distribution is 
\[
A_{i,j}  \sim \sum_{q=1}^Q  \sum_{l=1}^Q  \pi_q  \pi_l \{  \beta_{q,l}
\delta_{A^{(q,l)}} +(1-\beta_{q,l}) \delta_{0} \}.
\]
Fixing  three   distinct  integers  $i,j,k$  in   $\{1,\dots,n\}$  and
considering        the        trivariate       counting        process
$(N_{i,j},N_{i,k},N_{j,k})$  we end  up with  the distribution  of the
triplet  of random  measures  $(A_{i,j},A_{i,k},A_{j,k})$. There  the
expressions become more cumbersome but the very same reasoning may be
applied to identify   the  measures
$\{A^{(q,l)} : q,l=1,\dots, Q ;\ q\leq l\}$ up to a permutation
in $\mathfrak{S}_Q$. 
Concerning identification of  $\pi$ and $\beta$, we obtain  the set of
weights $\{ \pi_q^3\beta_{q,q}^3 : q=1,\dots, Q\}$ that is attached to
the $Q$ components corresponding to Dirac masses at points of the form
$(A^{(q,q)},A^{(q,q)},A^{(q,q)})$. Moreover we also obtain the set of weights
$\{  \pi_q^3\beta_{q,q}^2(1-\beta_{q,q}) :  q=1,\dots, Q\}$  attached to
Dirac masses at points $(A^{(q,q)},A^{(q,q)},0)$.  As the
$A^{(q,q)}$'s are unique we can  match the value $\pi_q^3\beta_{q,q}^3 $
with $\pi_q^3\beta_{q,q}^2(1-\beta_{q,q})  $ and thus  obtain (through
a simple ratio) $\beta_{q,q}$ and also $\pi_q$.  In other words the sets $\{\beta_{q,q} :
q=1,\dots, Q\}_q$ and $\{\pi_q : q=1,\dots, Q\}_q$ are identifiable. 
 Finally we may look at the weights $\pi_q^2\pi_l\beta_{q,q}\beta_{q,l}^2$
 associated with Dirac masses at points of the
 form  $(A^{(q,q)},A^{(q,l)},A^{(q,l)})$   with  $q\neq  l$.   As  the
 cumulative intensities  $\{A^{(q,l)} :  q,l=1,\dots, Q ;\  q\leq l\}$
 have been identified up to a permutation
in $\mathfrak{S}_Q$, together with the pair of $(\pi_q, \beta_{q,q})$'s (for the
same permutation), we obtain the values  $\{\beta_{q,l} :  q\neq  l\}$.  \\

Let us turn to inference of this model. To fix the notation we use the directed setup but similar equations may be
derived in the undirected case. 
The complete-data likelihood is 
\begin{align*}
\mathcal L_{\textrm{sparse}} (\mathcal O,\mathcal Z, \mathcal{U} \mid \theta)
=&\mathcal   L(\mathcal   O\mid\mathcal  Z,\mathcal{U},\theta)   \times
\mathcal L(\mathcal{U} \mid \mathcal Z, \theta) 
 \times \mathcal L(\mathcal Z\mid\theta) \\
=&\exp\left\{-\sum_{(i,j)\in\mathcal
    R}                                         U_{i,j}A^{(Z_i,Z_j)}(T)
  \right\} \times \left[\prod_{m=1}^M\alpha^{(Z_{i_m},Z_{j_m})}(t_m) \right] \\
& \times \left[\prod_{(i,j)\in \mathcal{R}} \prod_{q=1}^Q 
\prod_{l=1}^Q \left\{ \beta_{q,l}^{U_{i,j}}(1-\beta_{q,l})^{1-U_{i,j}} \right\}^{Z^{i,q} Z^{j,l}} \right] \times
\left[\prod_{i=1}^n\prod_{q=1}^Q\pi_q^{Z^{i,q}} \right].
\end{align*}

The true conditional distribution of the latent variables $(\mathcal Z,
\mathcal U)$ given the observations writes 
\begin{align*}
\Pt(\mathcal Z,\mathcal U \mid \mathcal O) = \Pt(\mathcal Z \mid \mathcal O)
\Pt(\mathcal U \mid  \mathcal Z,\mathcal O) = \Pt(\mathcal Z \mid \mathcal O) \prod_{(i,j)\in \mathcal{R}} \Pt (U_{i,j} \mid Z_i,
  Z_j , N_{i,j}) .
\end{align*}
A  main  difference  with  the  previous  setting  is  that  now  this
conditional distribution has two parts: the one concerning $\mathcal{U}$
has  a factorised  form and  can thus  be computed  exactly, while  the
 part concerning $\mathcal Z$  still has an intricate dependence
structure and we rely on a variational approximation to deal with it. 
We thus introduce a new conditional factorised
distribution      $\Pv(\cdot \mid \O)$      on     the      variables
$\mathcal{Z},\mathcal{U}$ that depends on the observations $\O$ and is defined
as 
\begin{align*}
&\Pv            \left(\mathcal            Z=(q_1,\dots,q_n),
  \mathcal{U}=(u_{i,j})_{(i,j)\in \mathcal{R}} \mid\mathcal O\right) \\
&=\prod_{i=1}^n\P_{\tau} (Z_i=q_i\mid\mathcal  O) \prod_{(i,j)\in
  \mathcal{R}}\Pt  (U_{i,j}=u_{i,j}\mid   Z_i=q_i,  Z_j=q_j
  ,N_{i,j}) \\
&=\prod_{i=1}^n\tau^{i,q_i} \prod_{(i,j)\in
  \mathcal{R}}\Pt  (U_{i,j}=u_{i,j}\mid   Z_i=q_i,  Z_j=q_j
  ,N_{i,j}) , 
\end{align*}
for any $ (q_1,\dots,q_n) \in\{1,\dots,Q\}^n$ and $ (u_{i,j})_{(i,j)\in
  \mathcal{R}} \in \{0,1\}^r$. 
As it does not depend on $\theta$, we let $\Pvtau(\Z \mid \O)$ denote the marginal distribution on $\Z$ of the
distribution $\Pv(\cdot \mid \O)$
and $\Evtau(\cdot \mid \O)$ the corresponding expectation. 
Moreover,  the true conditional distribution of $U_{i,j}$ is given by 
\begin{align}\label{eq:rho}
   &\Pt   (U_{i,j}=1  \mid   Z_i=q,  Z_j=l   ,  N_{i,j})   =
  1\{N_{i,j}(T) >0\} + \rho_{\theta}(q,l)  1\{N_{i,j}(T) =0\} := \rho_\theta(i,j,q,l)
     ,\nonumber \\
&\text{where } 
\rho_{\theta}(q,l) = \frac{\beta_{q,l}  \exp\{-A^{(q,l)}(T)\} } {1-\beta_{q,l}
             +\beta_{q,l} \exp\{-A^{(q,l)}(T)\}} . 
\end{align}
Indeed,  whenever  we observe  an  interaction  event between  $(i,j)$
(namely    $N_{i,j}(T)>0$)   we    know    that   $U_{ij}=1$    almost
surely. Otherwise  ($N_{i,j}(T)=0$), we  either have a  null intensity
process or a non-null intensity process with zero observations. 
Note that the parameters $\rho_{\theta}(q,l)$ 
 (or equivalently the $\rho_\theta(i,j,q,l)$) 
are not additional variational parameters; these are just functions of the original
parameter $\theta$. Finally we have 
\[\Pv            \left(\Z,  \U \mid\mathcal
  O\right) = \Big\{\prod_{i=1}^n \Pvtau(Z_i \mid \O ) \Big\} \times \prod_{(i,j)\in
  \mathcal{R}} {\rho_\theta(i,j,Z_i,Z_j)}^{U_{i,j}} ({1-\rho_\theta(i,j,Z_i,Z_j)})^{1-U_{i,j}}. 
\]

Let us now derive our variational approximation. 
Denoting by $\Ev(\cdot \mid \O)$ the expectation under the
distribution $\Pv(\cdot \mid \O)$ on $(\Z,\mathcal{U})$ and by
$\theta^{[s]}$ the current parameter value, we write  as usual 
\begin{align*}
 & \log \Ptp(\O) \\
  &= \Ev \{ \log \Ptp(\O,\Z,\U) \mid \O\} - \Ev \{\log \Ptp(\Z,\U \mid \O) \mid \O \} \\
&= \Ev \{ \log \Ptp(\O,\Z,\U)\mid \O \} + \mathcal{H}\{\Pv(\Z,\U \mid \O) \} + \KL{\Pv(\Z,\U \mid \O)}{\Ptp(\Z,\U \mid \O)}  .
\end{align*}

As a consequence, we introduce  a new criterion $\tilde {J}(\tau, \theta; \theta^{[s]})$ that is a lower bound on the
log-likelihood $ \log \Ptp(\O)$ and defined as

\begin{align*}
 \tilde{J} (\tau, \theta;\theta^{[s]}) =& \Ev \{ \log \Ptp(\O,\Z,\U) \mid \O \}
                  + \mathcal{H}\{ \Pv (\Z, \U \mid
  \mathcal{O}) \} \\
= &-\sum_{(i,j)\in \mathcal{R}} \sum_{q=1}^Q\sum_{l=1}^Q \tau^{i,q}\tau^{j,l} 
 \rho_{\theta}(i,j,q,l)  \{A^{[s]}\}^{(q,l)}(T)  \\
 & + \sum_{q=1}^Q\sum_{l=1}^Q  \sum_{m=1}^M   \tau^{i_m,q}\tau^{j_m,l}
   \log\left [ \{\alpha^{[s]}\}^{(q,l)}(t_m)\right] \\
  & + \sum_{(i,j)\in     \mathcal{R}}\sum_{q=1}^Q\sum_{l=1}^Q 
                                   \tau^{i,q}\tau^{j,l}
   \left[ \rho_{\theta}(i,j,q,l)\log \beta^{[s]}_{q,l} + \{1-\rho_{\theta}(i,j,q,l) \} \log (1-\beta^{[s]}_{q,l})
    \right ] \\
 &+\sum_{i=1}^n\sum_{q=1}^Q \tau^{i,q} \log \left\{\frac{\pi^{[s]}_q}
        {\tau^{i,q}} \right\}
        -  \sum_{(i,j)\in \mathcal{R}} \sum_{q=1}^Q\sum_{l=1}^Q \tau^{i,q}\tau^{j,l} \psi (\rho_{\theta}(i,j,q,l)), 
\end{align*}
where $\psi(\rho)= \rho \log \rho +(1-\rho)\log(1-\rho)$ is the entropy of the Bernoulli distribution with parameter
$\rho$. 
Using the definition of $\rho_{\theta}(i,j,q,l)$ and $\psi(1)=0$, the last
  term in the right-hand side   simplifies to 
 \[
\sum_{(i,j)\in \mathcal{R}} \sum_{q=1}^Q\sum_{l=1}^Q \tau^{i,q}\tau^{j,l} \psi (\rho_{\theta}(i,j,q,l))=  \sum_{q=1}^Q\sum_{l=1}^Q \psi(\rho_{\theta}(q,l)) \sum_{(i,j)\in \mathcal{R}} \tau^{i,q}\tau^{j,l}1\{N_{i,j}(T) =0\} .    
  \]
The variational \texttt{E}-step consists in maximizing $ \tilde{J}(\tau,\theta; \theta^{[s]})$ with respect to $(\tau,\theta)$. This is equivalent to 
choosing the variational distribution $\Pv$ that minimises the Kullback-Leibler divergence $ \KL{\Pv(\Z,\U \mid
  \O)}{\text{pr}_{\theta^{[s]}}(\Z,\U \mid \O)} $. The solution in $\theta$ is naturally obtained for
$\theta=\theta^{[s]}$. 

We need to choose the variational parameter $\tau$ that maximizes $
\tilde{J} (\tau, \theta^{[s]};\theta^{[s]})$.
Similarly to the non
sparse setup we obtain that $\tau^{[s]} $ satisfies a fixed point equation in $\tau$,
\begin{equation}
  \label{eq:fixie}
 \tau^{i,q}\propto \pi_q^{[s]} \exp\{\tilde{D}_{iq}( \tau,\theta^{[s]})\}, \quad
(i=1,\dots,n;\ q=1,\dots,Q), 
\end{equation}
where  
\begin{align*}
  \tilde{D}_{iq}(\tau,\theta) = &-\sum_{l=1}^Q \sum_{\substack{j=1\\
  j\neq i}}^n \tau^{j,l}
\left\{ \rho_\theta(i,j,q,l) A^{(q,l)}(T) + \rho_\theta(j,i,l,q) A^{(l,q)}(T) 
  \right\} \\
& 
-\sum_{l=1}^Q \psi(\rho_\theta(q,l) ) \sum_{\substack{j=1\\
  j\neq i}}^n \tau^{j,l} 1\{N_{i,j}(T) =0\} - \sum_{l=1}^Q\psi(\rho_\theta(l,q) ) \sum_{\substack{j=1\\
  j\neq i}}^n \tau^{j,l}  1\{N_{j,i}(T) =0\} 
\\
&+\sum_{l=1}^Q\sum_{m=1}^M    \left[   \1_{\{i_m=i\}}\tau^{j_m,l}
  \log\left\{\alpha^{(q,l)}(t_m)\right\}+\1_{\{j_m=i\}}\tau^{i_m,l}\log\left\{\alpha^{(l,q)}(t_m)\right\}
  \right] \\
&   +  \sum_{l=1}^Q   \sum_{\substack{j=1\\  j\neq   i}}^n  \tau^{j,l}
  [  \rho_\theta(i,j,q,l)\log  \beta_{q,l}   +  \{1-\rho_\theta(i,j,q,l)  \}  \log
  (1-\beta_{q,l}) \\
& \qquad \qquad +\rho_\theta(j,i,l,q)\log  \beta_{l,q}   +  \{1-\rho_\theta(j,i,l,q)  \}  \log
  (1-\beta_{l,q}) ] .
\end{align*}

The  \texttt{M}-step consists in maximizing $\tilde{J}(\tau^{[s]},\theta^{[s]}; \theta)$ with respect to $\theta$. It 
is  again divided  into two  parts, treating  the
finite-dimensional parameter $(\pi,\beta)$ differently than the infinite
dimensional     one    $\alpha$.     We     thus    first     maximize
$\tilde{J} (\tau^{[s]},\theta^{[s]};\pi, \beta,\alpha)$ with respect to $(\pi,\beta)$ using the current parameter value
$\alpha=\alpha^{[s]}$. 
The solution with  respect to $\pi$ is  the same as in  the non sparse
case and given in~\eqref{eq_mstep_pi_q}.  
Now optimization with respect to $\beta$ leads to (denoting $
\rho^{[s]}= \rho_{\theta^{[s]}}$), 
\begin{equation}
  \label{eq:betaql}
  \beta_{q,l}^{[s+1]} = \frac{\sum_{(i,j)\in  \mathcal R} \{\tau^{[s]}\}^{i,q}\{\tau^{[s]}\}^{j,l}
    \rho^{[s]}(i,j,q,l)}{ \sum_{(i,j)\in  \mathcal R} \{\tau^{[s]}\}^{i,q}\{\tau^{[s]}\}^{j,l}}
  , \quad q,l=1,\dots, Q. 
\end{equation}

Then estimation of the intensities $\alpha^{(q,l)}$ is done exactly as
previously, except that we replace the variational process $N^{(q,l)}$
by      $\tilde{N}^{(q,l)}      =     \sum_{(i,j)\in      \mathcal{R}}
\rho^{[s]}(i,j,q,l)\{\tau^{[s]}\}^{i,q}\{\tau^{[s]}\}^{j,l} N_{i,j}$. \\

Finally we  start from  an initial value  of the  clusters $\Z$
  (see   Section~\ref{sec:init})  that   we  treat   as  probabilities
  $\{(\tau^{i,q})_{1\le q \le Q}; 1\le i \le n\}$.  Then we initialise the sparsity parameters
  $\beta_{q,l}$ and mean intensities $A^{(q,l)}(T)$ with 
\[
\beta_{q,l}   =\frac    {\sum_{(i,j)\in   \mathcal{R}}   Z^{i,q}
  Z^{j,l}1\{N_{i,j}(T)>0\} } { \sum_{(i,j)\in \mathcal{R}} Z^{i,q}
  Z^{j,l}} ,\qquad
A^{(q,l)} (T) = \frac{\sum_{(i,j)\in   \mathcal{R}}   Z^{i,q}
  Z^{j,l} N_{i,j}(T)} {\sum_{(i,j)\in   \mathcal{R}}   Z^{i,q}
  Z^{j,l}1\{N_{i,j}(T)>0\} }.
\]
This enables to initialise $\rho(i,j,q,l)$ with~\eqref{eq:rho}. 
 After these initialisations, we are ready to iterate the following steps. At iteration $s\ge 1$
we do 
\begin{itemize}
\item    \texttt{M}-step:    Update     $\pi^{[s+1]}$    via~\eqref{eq_mstep_pi_q} with  $\tau^{[s]}$;  Update $\beta^{[s+1]}$  via~\eqref{eq:betaql} with   $\tau^{[s]},\rho^{[s]}$; Update $\alpha^{[s+1]}$   either  via Equation~\eqref{eq:mstep_histo}
  (histogram   method)   or~\eqref{eq_mstep__kernel_lamb_ql}   (kernel
  method)  using  the  process $\tilde{N}^{(q,l)}  $  and  variational
  parameters $\rho^{[s]}$ and  $\tau^{[s]}$. 
\item \texttt{VE}-step: Update   the values   $\rho^{[s+1]}$ via~\eqref{eq:rho} with $\beta^{[s+1]}$ and
  $A^{[s+1]}(T)$ derived from $\alpha^{[s+1]}$; Update  $\tau^{[s+1]}$ as  the solution to the fixed
  point  equation~\eqref{eq:fixie}  relying   on  the  current  values
  $\pi^{[s+1]}, \rho^{[s+1]},\alpha^{[s+1]}$, $\beta^{[s+1]}$. 
\end{itemize}

The integrated classification likelihood criterion becomes
\begin{equation*}
\ICL_{\textrm{sparse}} (Q) = \log \mathbb{P}_{\hat \theta(Q)}\{ \mathcal{O},\hat
\tau(Q), \hat \rho(Q)\} -\frac 1 2 (Q-1)\log n - \frac 1 2 \log r \Big[ Q^2 +
\sum_{q=1}^Q\sum_{l=1}^Q 2^{\hat d^{(q,l)} } \Big].
\end{equation*}

\section{The sparse setup: examples}

We first discuss the results of the sparse analysis on the London bike sharing system dataset. 
First let us recall that in this dataset only $7\%$ of pairs of bike stations have at least one interaction. The main
model ignores that fact and this impacts the results in the sense that groups are mainly driven by these absences of
interactions. For instance the clusters obtained are mostly geographic, revealing absences of connections between
distant bikes stations (nonetheless we also discovered an interesting small cluster with this model). We thus decide to
explore whether one can decipher different structure with our sparse setup. In the following we focus on day 1 as similar results were obtained for day 2.

First our sparse integrated classification likelihood criterion selects only $\hat Q=2$ groups (compared to
$\hat Q=6$ in the non sparse case).  
 Geographic locations of the bike stations and the resulting clusters are represented on a city map (thanks to the
OpenStreetMap project), see Figure~\ref{fig:cycles_day1_sparse}. There is one
group containing a central part of the city while the remaining stations form a large peripheral cluster. Looking
at the estimated intensities (Figure~\ref{fig:cycles_day1_intens_sparse}) we see that the second group
(i.e. the central geographical group)  has large intra-group intensity with three modes: one in the morning
(around 8:30 am), one at lunch (around 1pm) and the last at the end of the day (at 5:50pm). Group 1 (the peripheral
one) mostly consists in `leaving' stations in the morning (see mode in the estimated intensity for $(q,l)=(1,2)$ around 8:20am)
and in `arriving' stations at the end of the day (mode in the estimated intensity for $(q,l)=(2,1)$ at 5:50pm).
Intra-group interactions in group 1 have a much lower intensity. On this dataset the sparse setup appears as a
complementary model that may shed some different light on the data. \\

We also analysed the Enron corpus with the sparse model as   $91$\% of the pairs of individuals  do not exchange any
email during the observation time. The sparse integrated classification likelihood criterion chooses $\hat Q=10$ as the optimal
number of groups, which is smaller than in the non sparse model where no optimum has been found in the  range of $Q$
from 1 to 20. As in the non sparse model the algorithm identifies one large group with 125 members, while the other nine
groups contain at most five individuals. {The adjusted rand index of the clustering in the sparse case with $Q=10$ (or $Q=4$)
  and the clustering in the non sparse model with $Q=4$ equals $0.51$ ($0.52$), which means that there are substantial
  differences between the  clusterings in the two models.} Figure \ref{enron_sparse_beta_10} shows the estimated values of the connectivity probabilities $\beta_{q,l}$. Most of these probabilities are significantly lower than $1$ justifying the application  of the sparse model to these data. 
A consequence of low connectivity probabilities $\beta_{q,l}$ is that the estimated intensities are more elevated than in the non sparse case which can be observed in Figure \ref{enron_sparse_alpha_Q10} in comparison to the intensity values obtained in the non sparse model (Figure \ref{enron_fig_heatmap}). We can also compare the estimated intensities in the sparse model with $Q=4$ (Figure~\ref{enron_sparse_intens_Q4}) with those in the non sparse case. Again we see that the intensities in the sparse setup are much more elevated. Moreover, the form of the intensities involving two small groups (i.e. $(q,l)\in\{2, 3 ,4\}^2$) are all quite different in the two models. 

We conclude that as in the London bikes example the results in the
sparse model differ much from those in the non sparse case. The sparse
model tends to select a smaller number of groups which makes
interpretation of results easier. As many real datasets are sparse in
the sense that only a small percentage of individuals effectively
interact with another the sparse model seems to be particularly adapted
to real data and provides the possibility of further insights on the
data. \\

We also analysed the primary school dataset with the sparse model. In this dataset $28\%$ of pairs of individuals have at least one interaction. 
Our sparse  integrated classification likelihood criterion selects $\hat Q=13$
groups, which is smaller than in the non sparse model. The clustering in
the non sparse model and in the sparse model are quite close, with some
groups being the same.  
The main difference between the two clusterings concerns the group
composed in the non sparse model of class 4B with additional  pupils
coming from almost all other classes. This group was characterized by
the lowest intra-group intensity. In the sparse model, class 4B is
separated into two groups: 6 pupils  are gathered with  class
4A to form one group (group 1), whereas the 17 remaining pupils are
gathered with class 1A and some pupils coming from almost all other
classes (group 3). Looking at the estimated intensities, we see that
group 3 has a low intra-group intensity during the lunch time contrary to
group 1 (see Figure~\ref{primaryschool_sparse_class_4B_1A}). Moreover the estimated intra-connectivity
probability for groups 1 and 3 are given by $\hat{\beta}_{1,1}=0.84$ and
$\hat{\beta}_{3,3}=0.29$.  Therefore group 3 is composed of individuals
which only a few proportion interacts, and characterized by very few
interactions during the lunch time.
On this dataset  with the non sparse and sparse models we mainly recover the
same clustering based on communities, but
the sparse model also exhibits particular temporal profile of some
individuals.

\newpage 

\begin{figure}
  \centering
  \includegraphics[width=\textwidth]{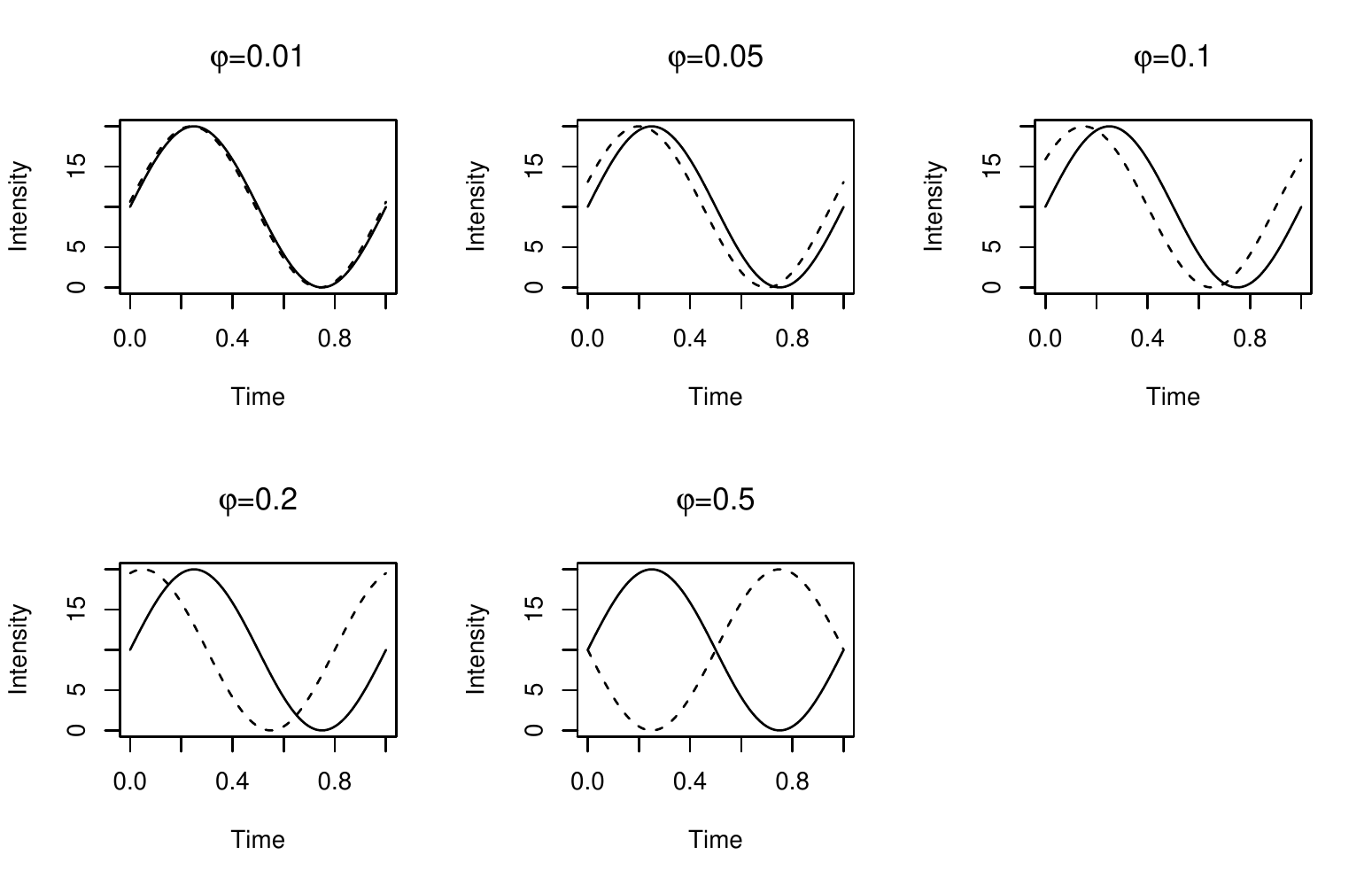}
  \caption{Intensities in synthetic experiments from Scenario 1. Each picture represents the intra-group intensity $\alpha^{\text{in}}$ (bold line) and the
    inter-group intensity  $\alpha^{\text{out}}$  (dotted line) with different
    shift parameter values $\varphi \in \{0.01,0.05,0.1,0.2,0.5\}$. }
  \label{fig_intensities_Q2}
\end{figure}

\begin{table}
\centering
\def~{\hphantom{0}}
\caption{Mean number of events and risks with standard deviations (sd) in scenario 2 with $n=20$.
Histogram (Hist) and kernel (Ker) estimators are compared with their oracle counterparts (Or.Hist, Or.Ker). 
All values associated with the risks are multiplied by 100.}
{
\begin{tabular}{lccccc}
 Groups $(q,l)$ & Nb.events & Hist (sd)  & Or.Hist (sd)  & Ker (sd) & Or.Ker (sd) \\ 
\\
 $(1,1)$ & 84 & 136 (92) & 50 (49) & 215 (83) & 113 (55) \\ 
   $(1,2)$ & 146 & 177 (146)  & 98 (27)  & 270 (107) & 194 (23) \\ 
     $(1,3)$  & 86 & 211 (160) & 78 (20) & 178 (143) & 43 (18) \\ 
    $(2,2)$  & 32 & 136 (72)& 108 (29)  & 139 (109) & 71 (41) \\ 
    $(2,3)$  & 130 &  265 (72) & 217 (28) & 238  (78) & 182  (22) \\ 
    $(3,3)$ & 48 & 173 (61) & 158 (47) & 171 (111) & 85 (43) 
\end{tabular}
}
\label{risk_Q3_n20}
\end{table}

\begin{figure}[h]
  \centering
  \includegraphics[width=\textwidth]{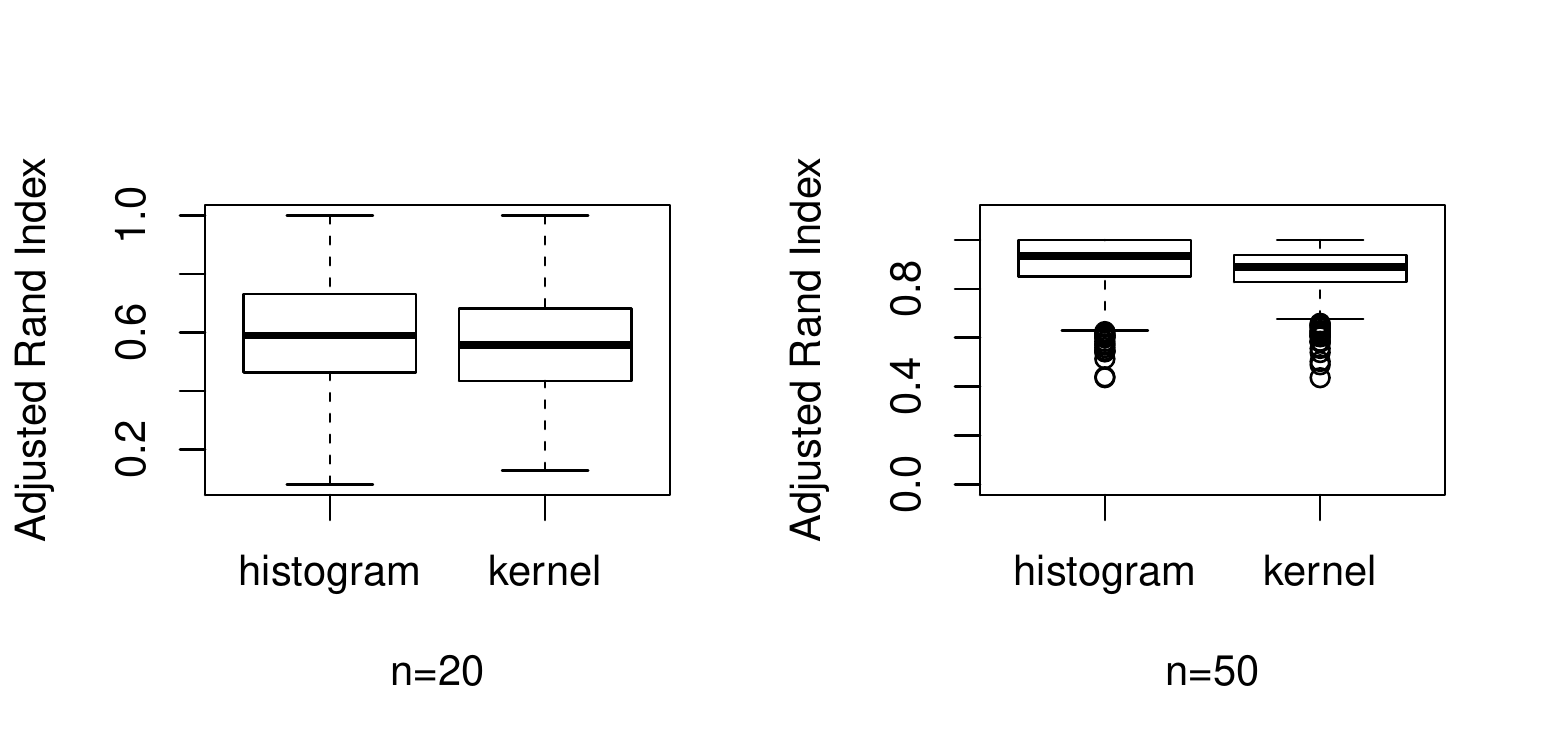}
  \caption{Boxplots   of  the   adjusted  rand   index  in   synthetic
    experiments from Scenario 2, for the histogram (left) and the kernel (right) estimators. Left panel: $n= 20$, right panel: $n=50$.}
\label{boxplot_Q3}
\end{figure}

\begin{figure}[h]
  \centering
 \includegraphics[width=\textwidth]{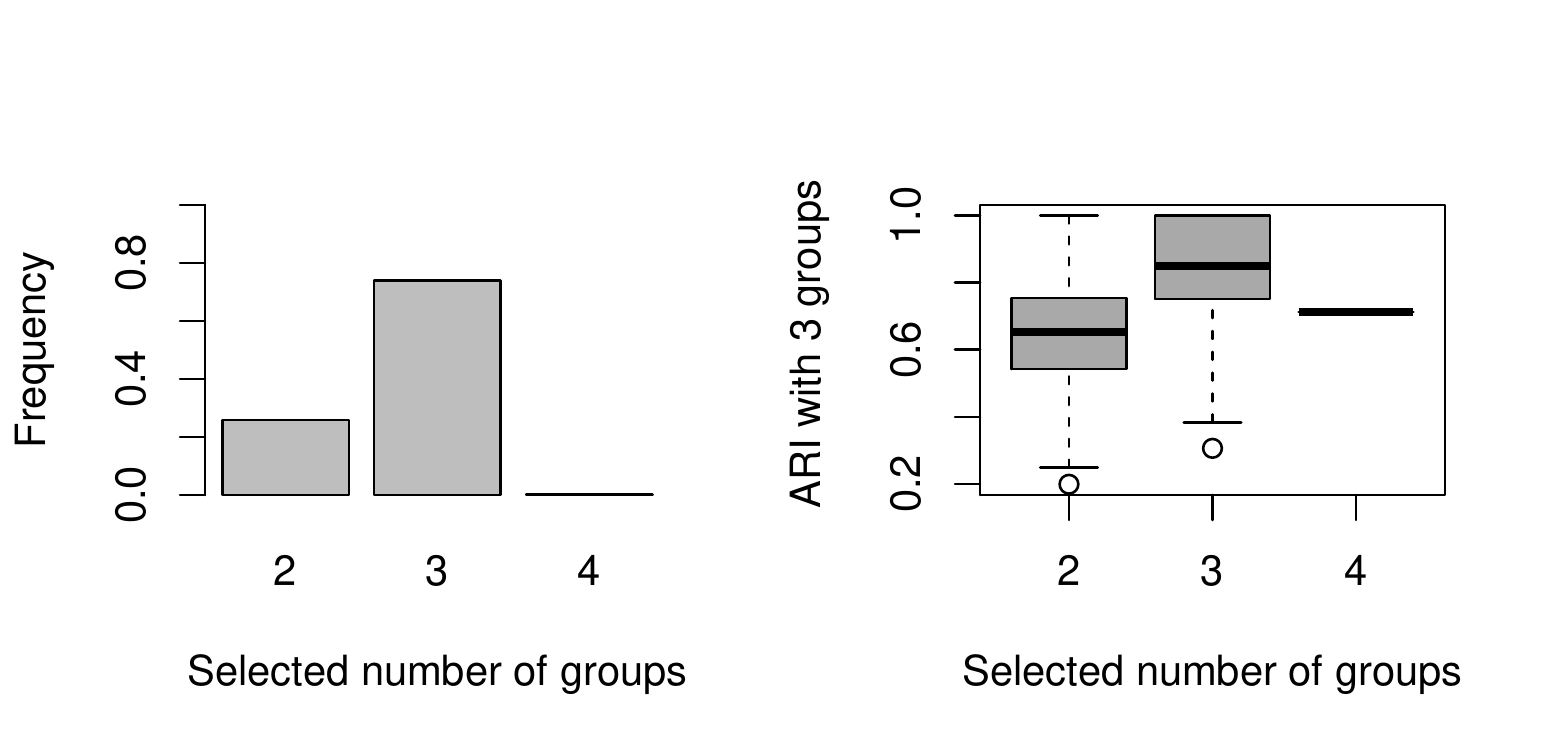}
  \caption{Selection of the number of latent groups
    via the integrated classification likelihood criterion in Scenario 2 with $n= 20$. Left panel:   frequencies of  selected
    number of groups. Right panel: adjusted rand index between the classification into 
    three groups and true classification as a function of the number of selected groups by \ICL (in    $\{2,3,4\}$).}
  \label{Qbest-20}
\end{figure}

\begin{figure}[h]
  \centering
 \includegraphics[width=\textwidth]{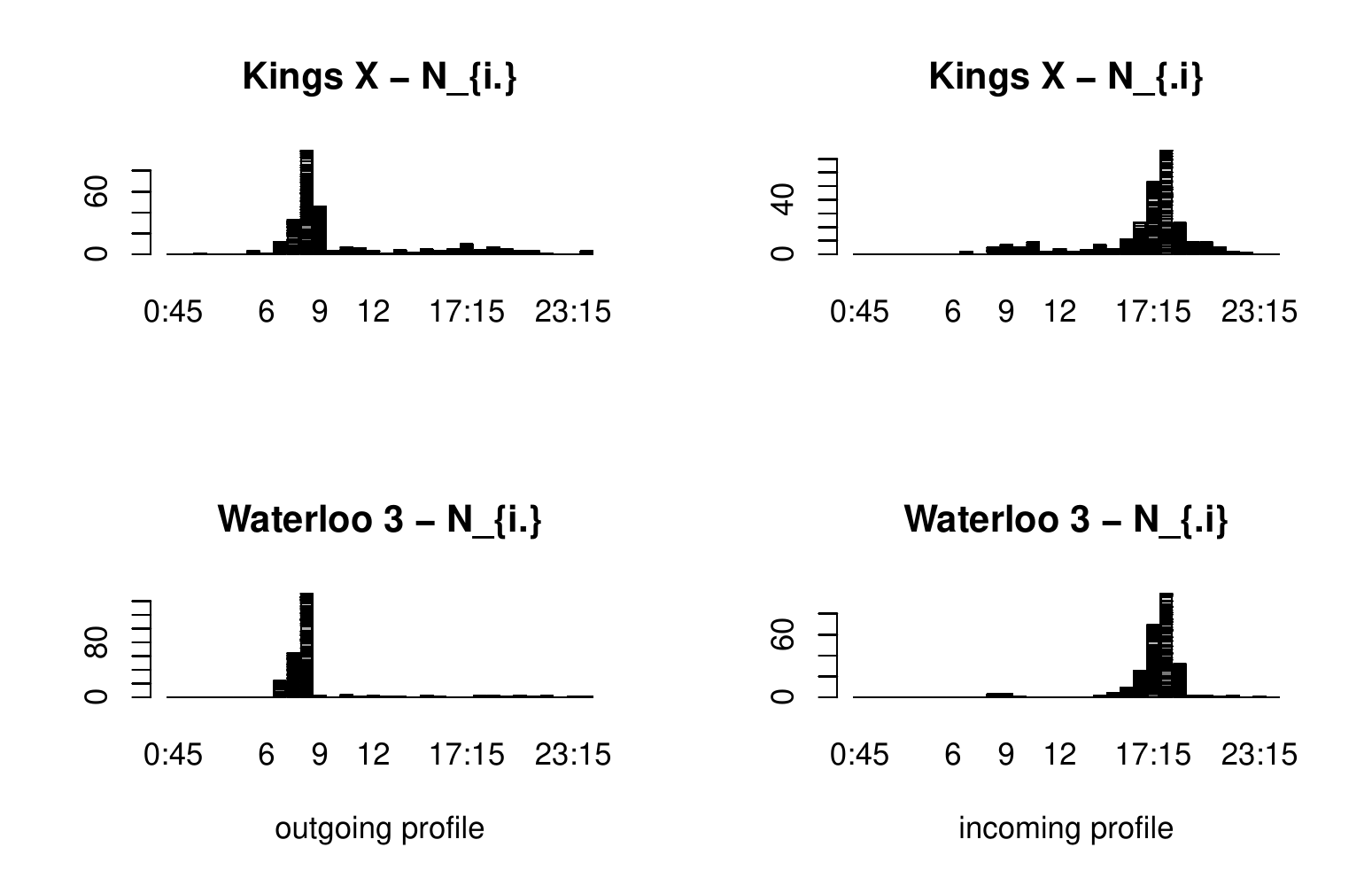}
  \caption{London bike sharing system: Barplots of outgoing ($N_{i\cdot}(\cdot)$ on the left) and incoming ($N_{\cdot
      i}(\cdot)$ on the right) processes from the 2 stations $i$ (top row and bottom row, respectively) in the small  cluster: representation of 
    volumes of connections to all other  stations during day 1 (time on the $x$-axis).}
  \label{fig:cycles_temp}
\end{figure}

\begin{figure}
  \centering
  \includegraphics[width=\textwidth]{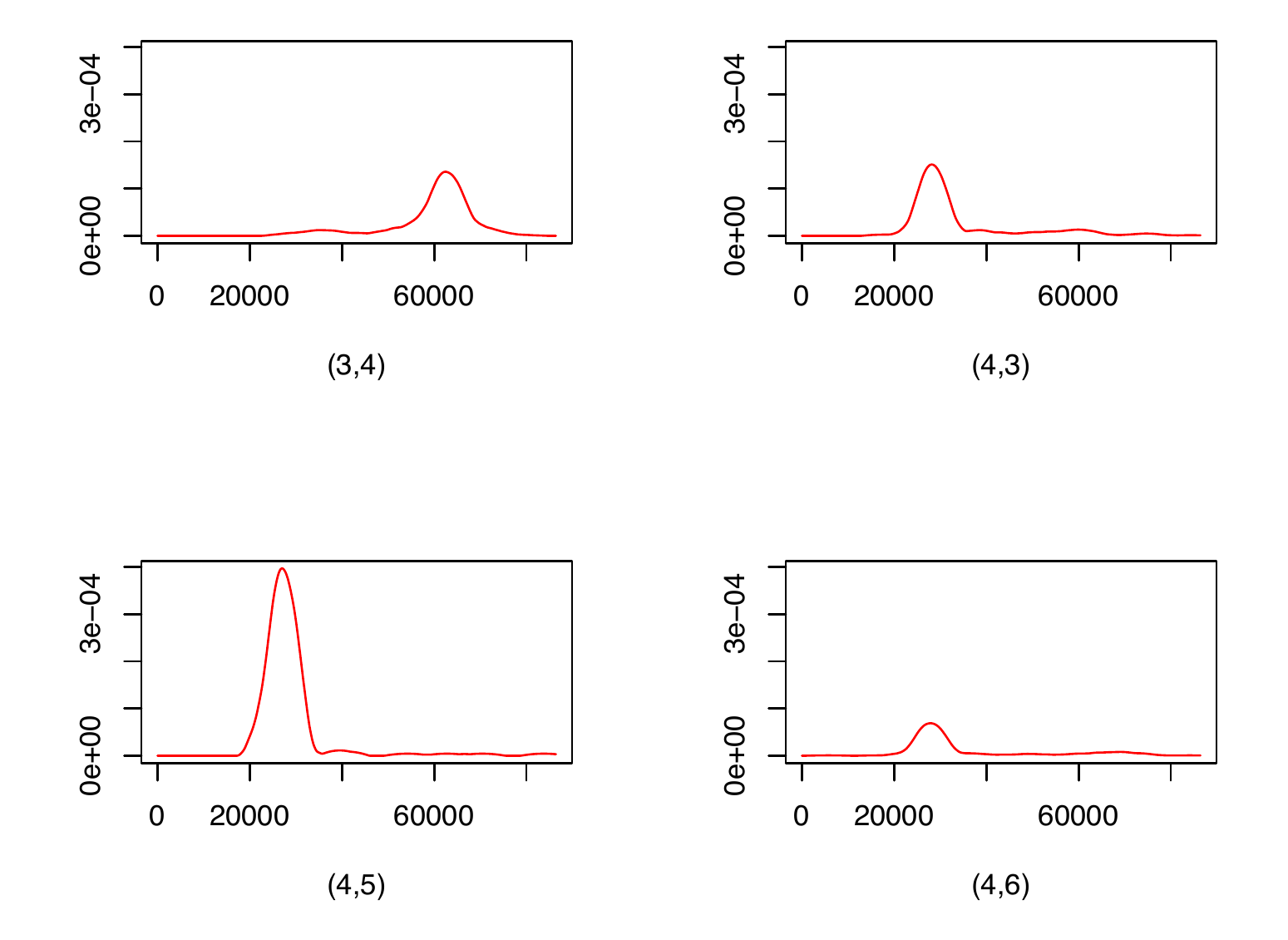}
  \includegraphics[width=\textwidth]{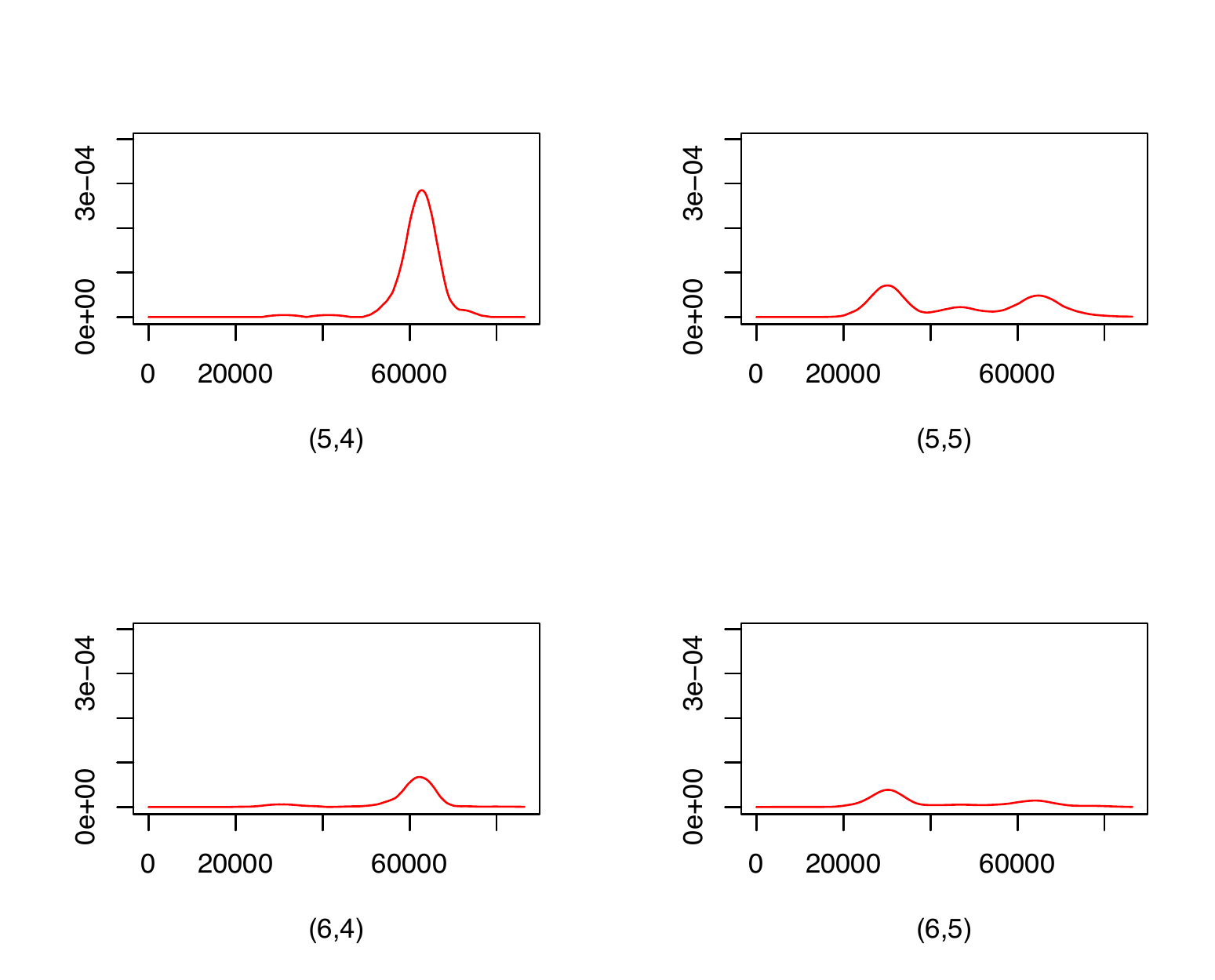}  
  \caption{London bike sharing system: estimated non almost null intensities for day 1 (time on the $x$-axis is in seconds).}
    \label{fig:cycles_intens}
\end{figure}

\begin{table}
\centering
\def~{\hphantom{0}}
\caption{Enron: Total size and group composition with $Q=4$ groups (some people's positions are unknown).}
{
\begin{tabular}{cccc}     
&total&managers&employees\\
\\
group 1&127&62&36\\
group 2&4&0&3\\
group 3&2&1&1\\
group 4&14&12&1
\end{tabular}
}
\label{tab_enron_groupcomposition}%
\end{table}

\begin{figure}[ht]
\centering
      \includegraphics[width=0.7\textwidth]{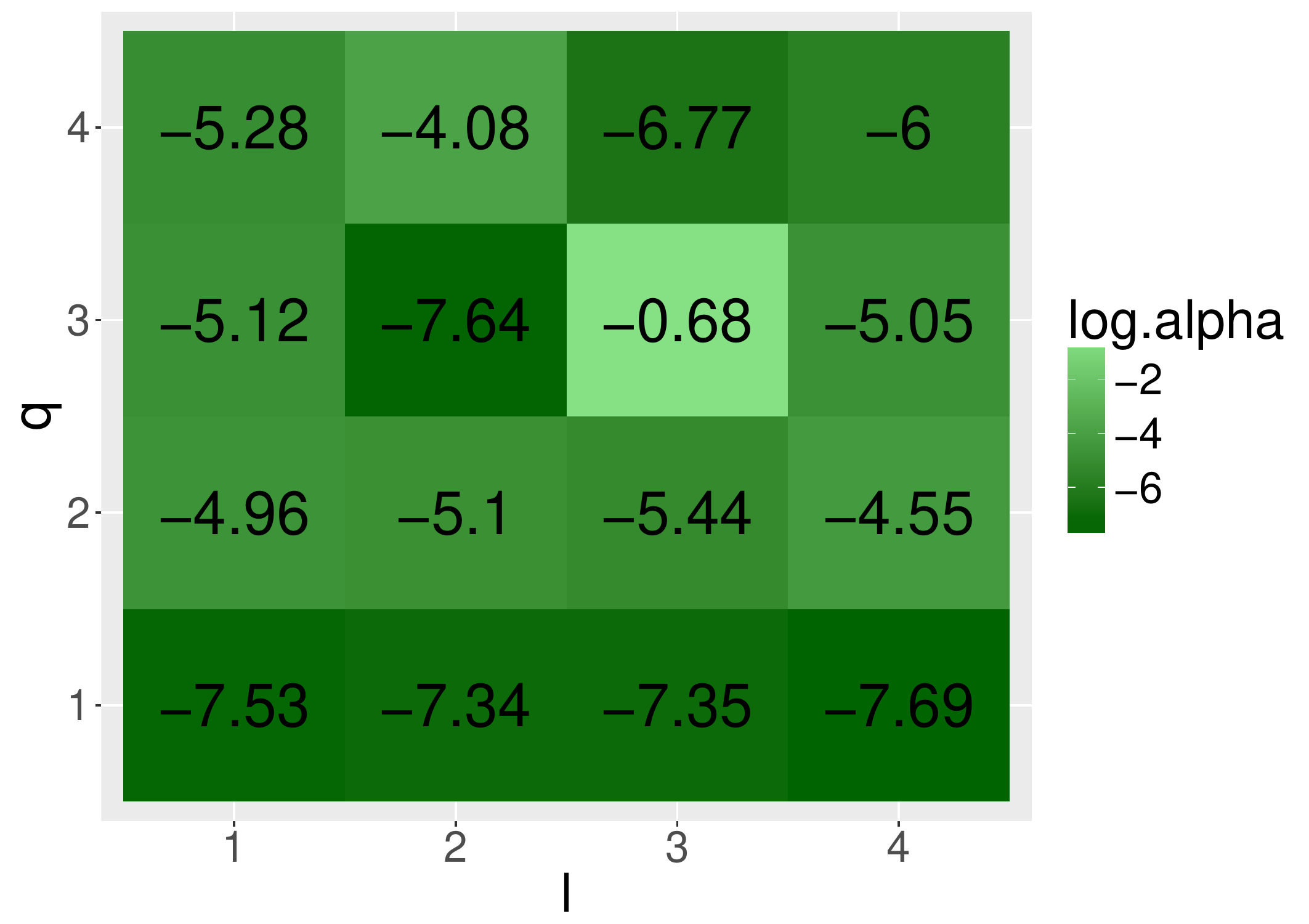}
  \caption{Enron: Logarithm of the mean values of the estimated intensities $\alpha^{(q,l)}$ with $Q=4$ groups.}
\label{enron_fig_heatmap}%
\end{figure}

\begin{landscape}
\begin{figure}
  \centering
    \includegraphics[width=\textwidth]{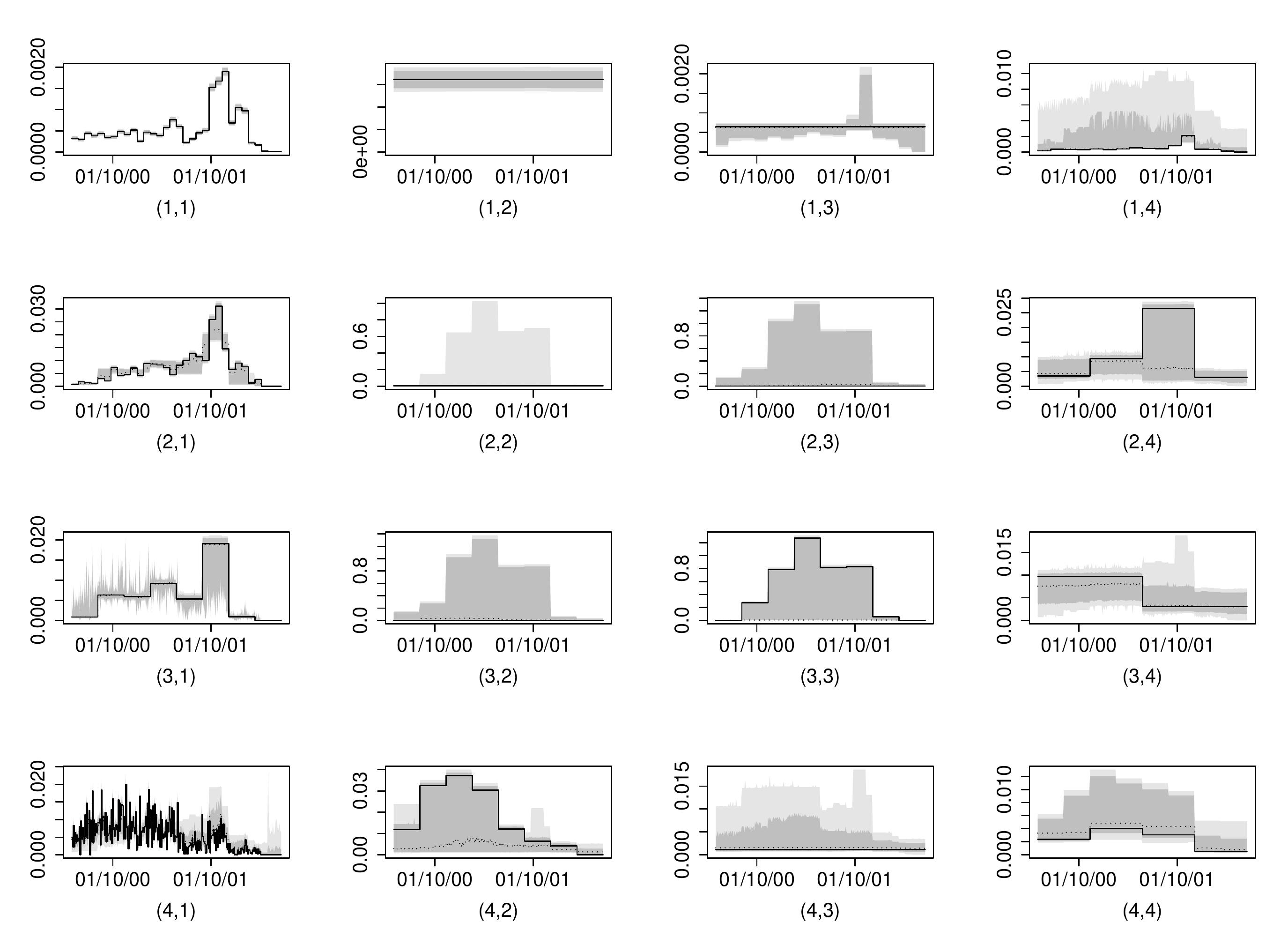}
  \caption{Enron: Estimated intensities  $\hat\alpha^{(q,l)}$ with $Q=4$ groups with bootstrap confidence intervals with confidence level 90\% (lightgrey) and 80\% (drakgrey) and the median bootstrap values (dotted lines).}
  \label{fig_enron_all_intensities}
\end{figure}
  \end{landscape}

\begin{figure}[h]
  \centering
  \includegraphics[width=\textwidth]{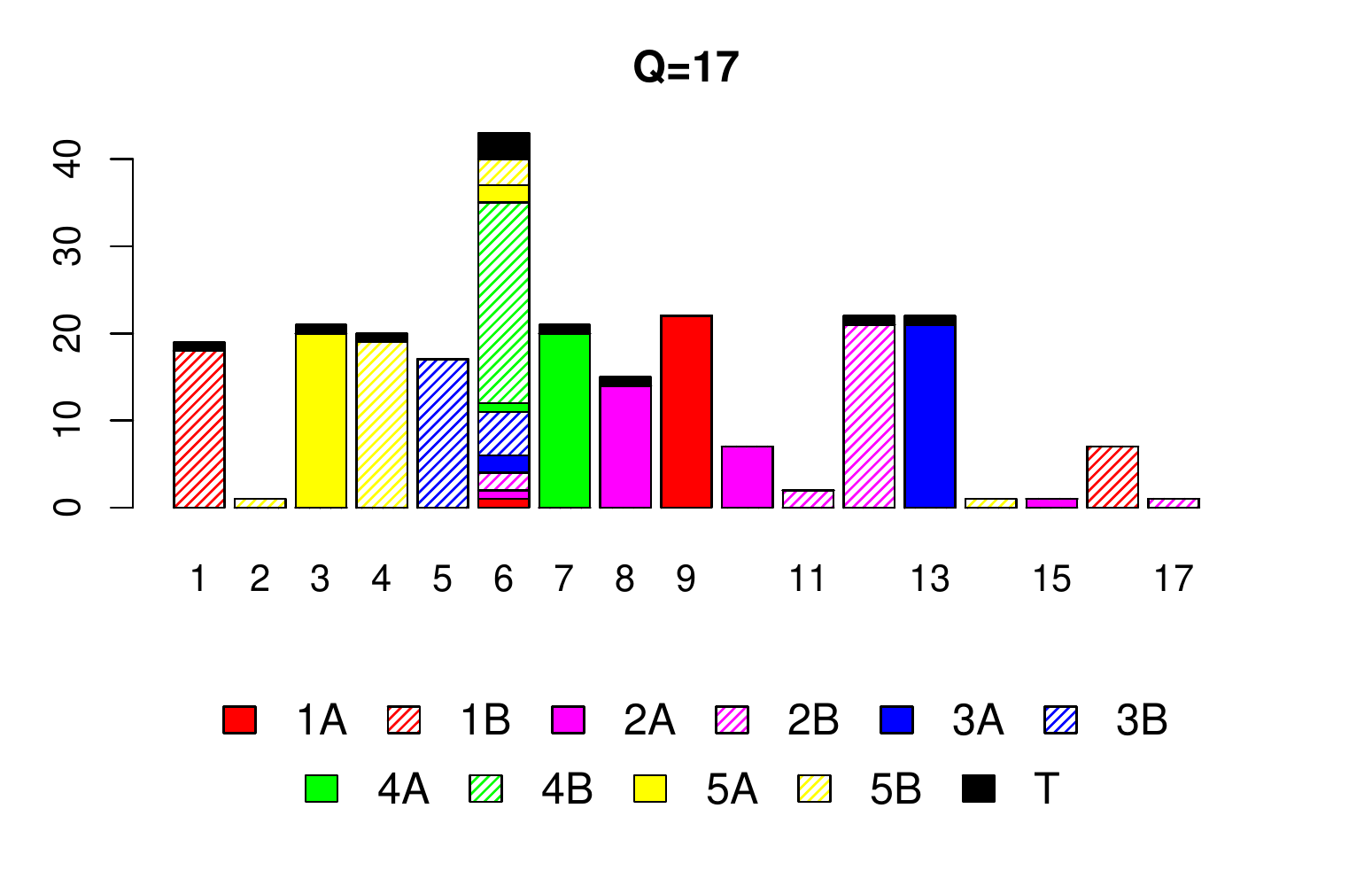}
   \caption{Primary school: clustering of the $242$ individuals  into $Q= 17$ groups.
    Vertical  bars    represent the   $Q$ clusters.
    Colours 
    indicate the   grades  and   the
    teachers, plain  and hatching  distinguish the  two classes  in the
    same  grade.}
  \label{primaryschool_tau}
\end{figure}

\begin{figure}[h]
  \centering
  \includegraphics[width=\textwidth]{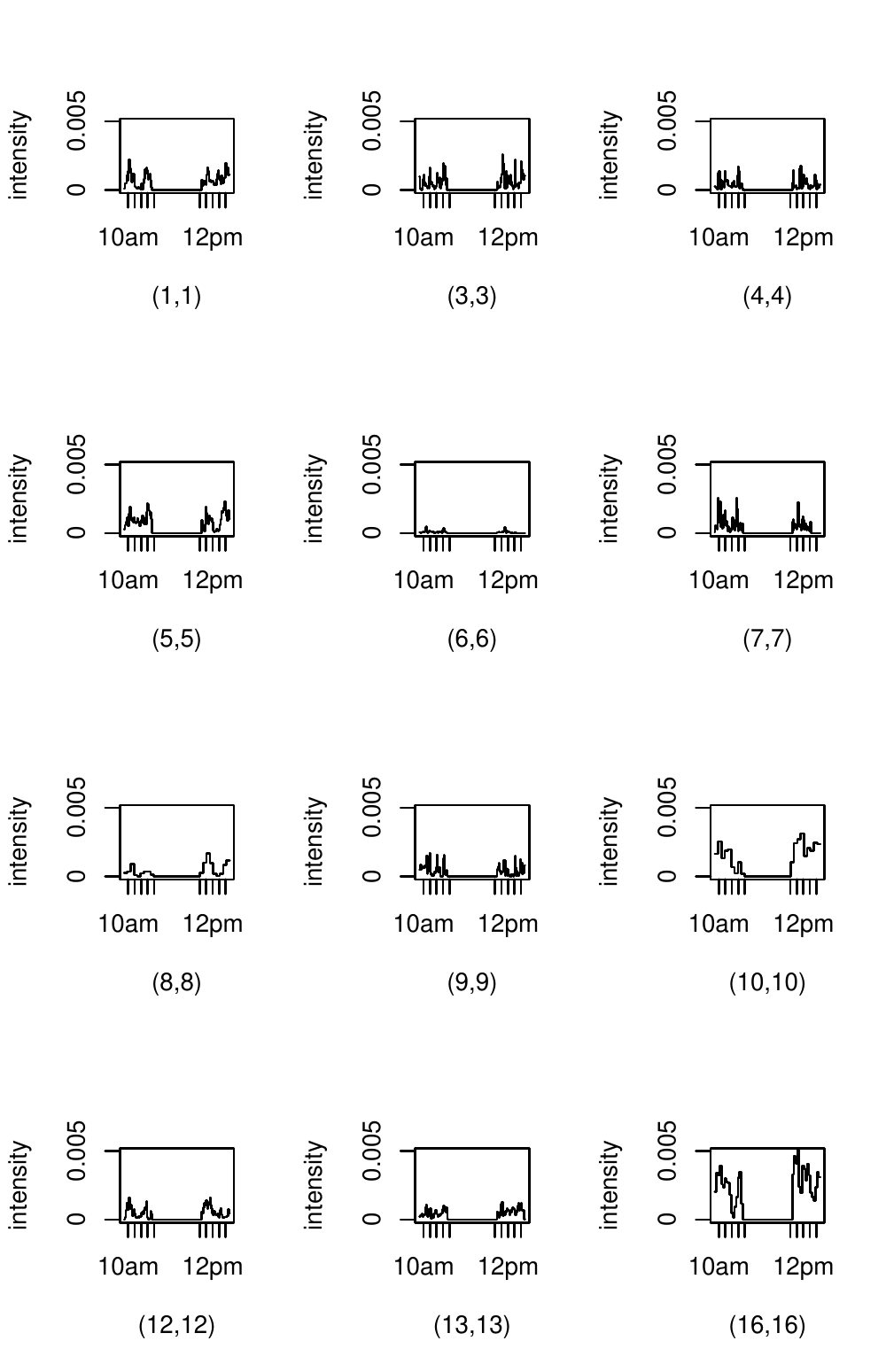}
   \caption{Primary school: Estimated intra-group intensities (plotted
     on the same $y$-scale).}
  \label{primaryschool_intra_group}
\end{figure}

\begin{figure}[h]
  \centering
  \includegraphics[width=\textwidth]{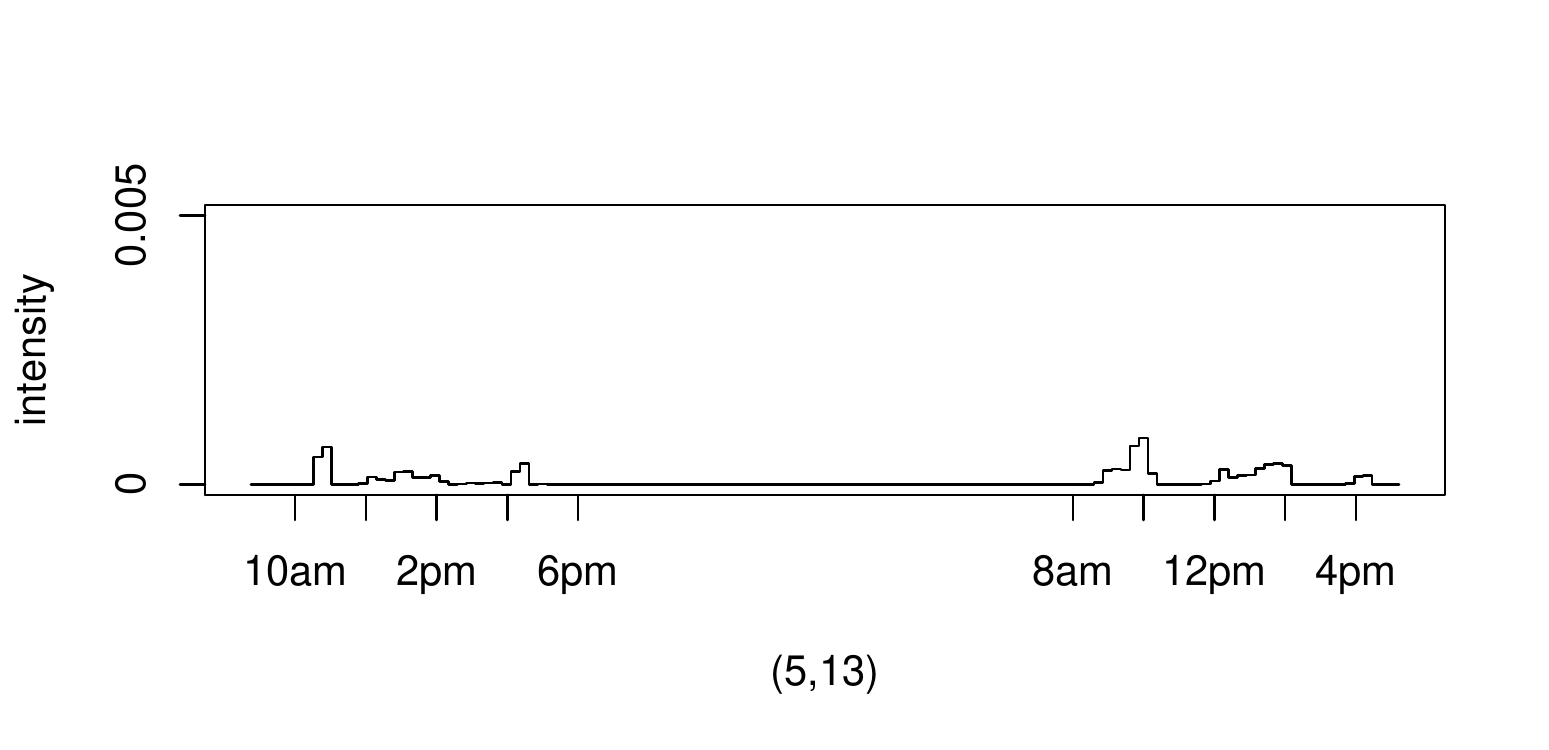}
   \caption{Primary school: Estimated inter-group intensity  between
      two classes of the same grade: classes $3A$ (group  13) and 
    $3B$ (group 5).}
  \label{primaryschool_samegrade}
\end{figure}

\begin{figure}[h]
  \centering
  \includegraphics[width=\textwidth]{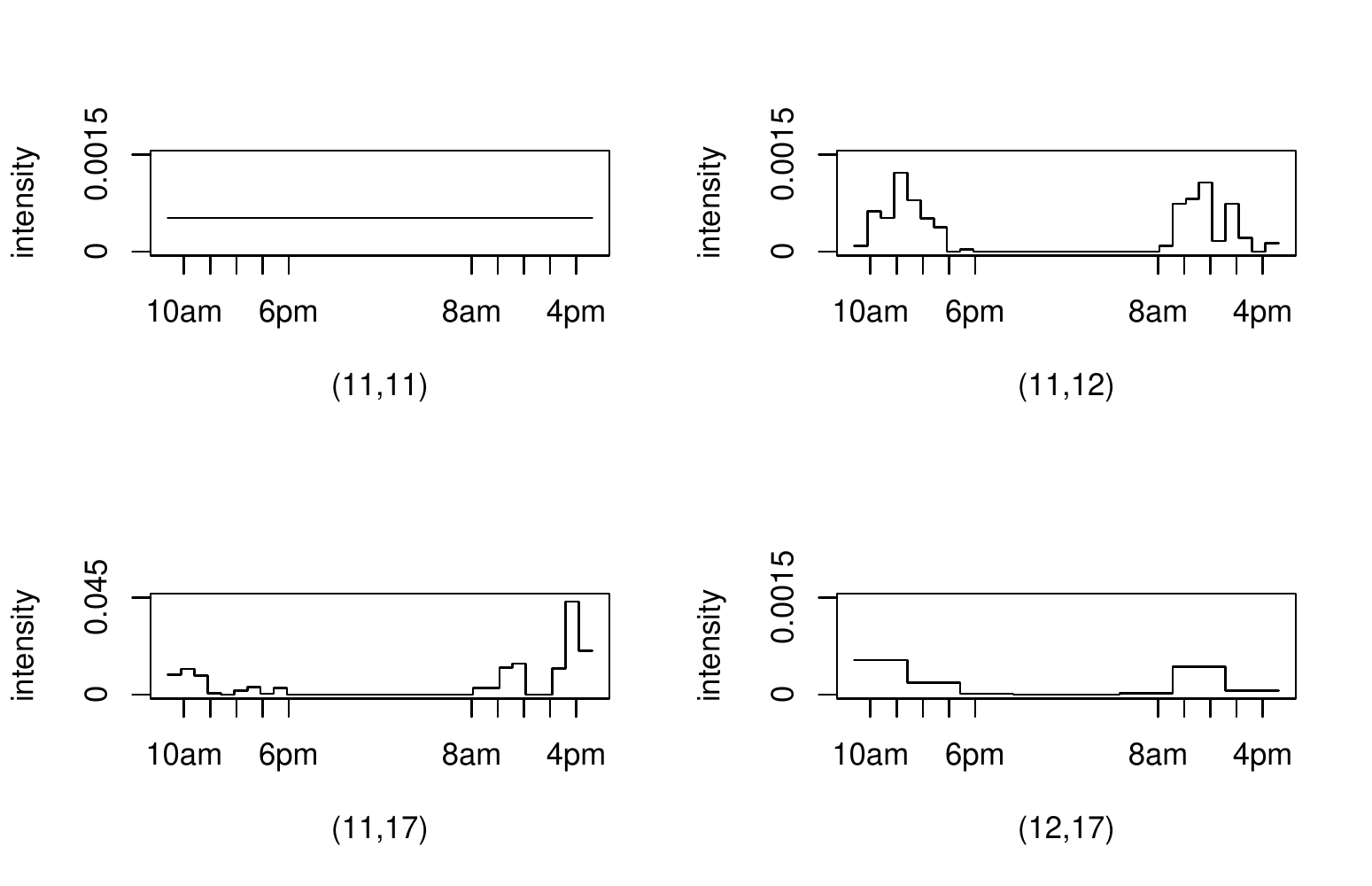}
  \caption{Primary school: Estimated intensities. 
    Example of class 2B splited into group 12 (with 21 pupils), group
    11 (with 2 pupils), and group 17 (with 
only one pupil).}
  \label{primaryschool_class_2B}
\end{figure}

\begin{figure}[h]
  \centering
 \includegraphics[width=\textwidth]{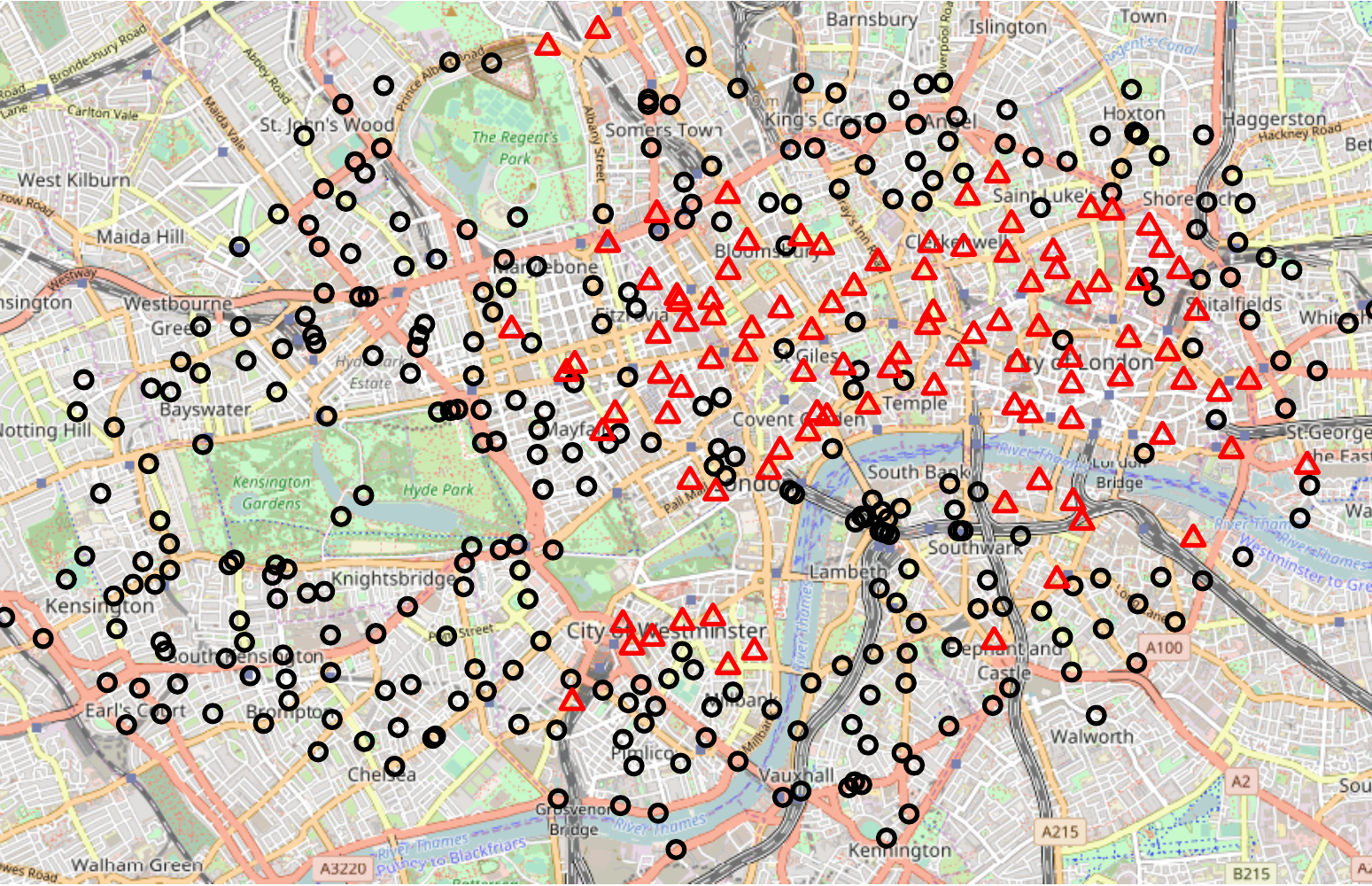}
  \caption{London bike sharing system: Geographic positions of the stations   
  and  clustering  into  two clusters (represented by  different  colors and symbols) obtained from the sparse model for day 1.}
  \label{fig:cycles_day1_sparse}
\end{figure}

\begin{figure}[h]
  \centering
  \includegraphics[width=\textwidth]{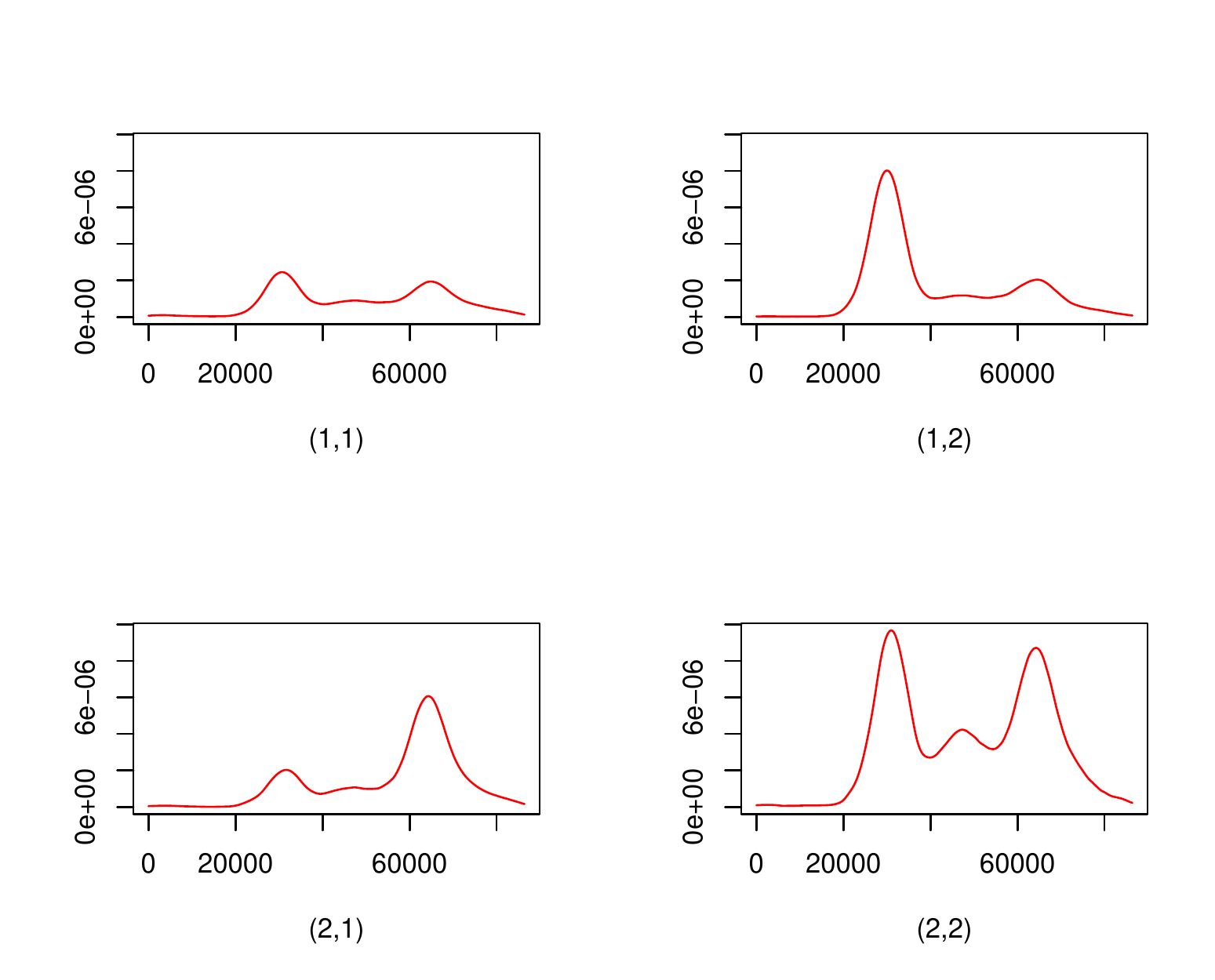}
  \caption{London bike sharing system: estimated intensities from the sparse model for day 1 (time on the $x$-axis is in seconds).}
    \label{fig:cycles_day1_intens_sparse}
\end{figure}

\begin{figure}[h]
  \centering
  \includegraphics[width=.7\textwidth]{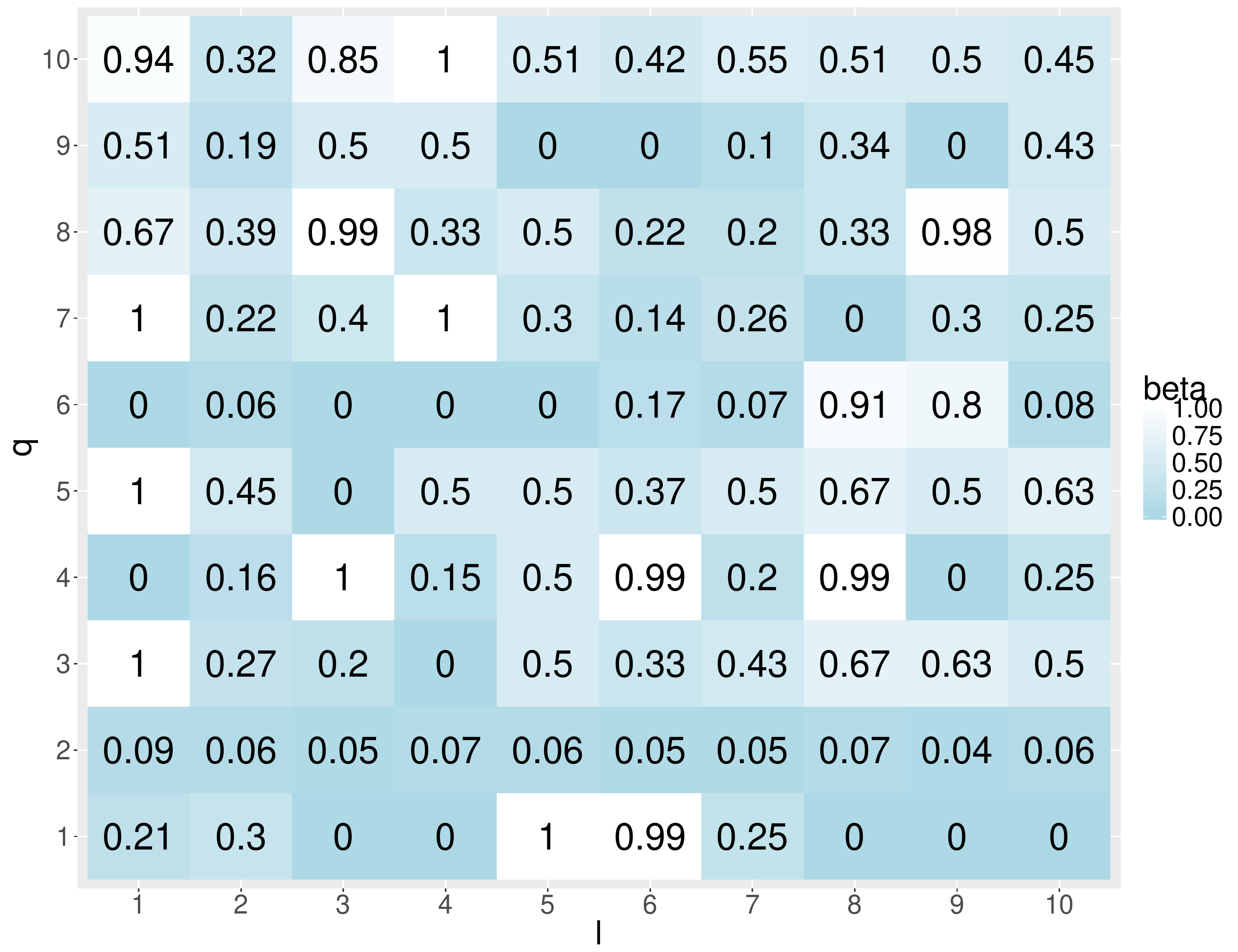}
  \caption{Enron: Estimated connectivity probabilities $\beta_{q,l}$ in the sparse model with $Q=10$ groups.}
    \label{enron_sparse_beta_10}
\end{figure}

\begin{figure}[h]
  \centering
  \includegraphics[width=.7\textwidth]{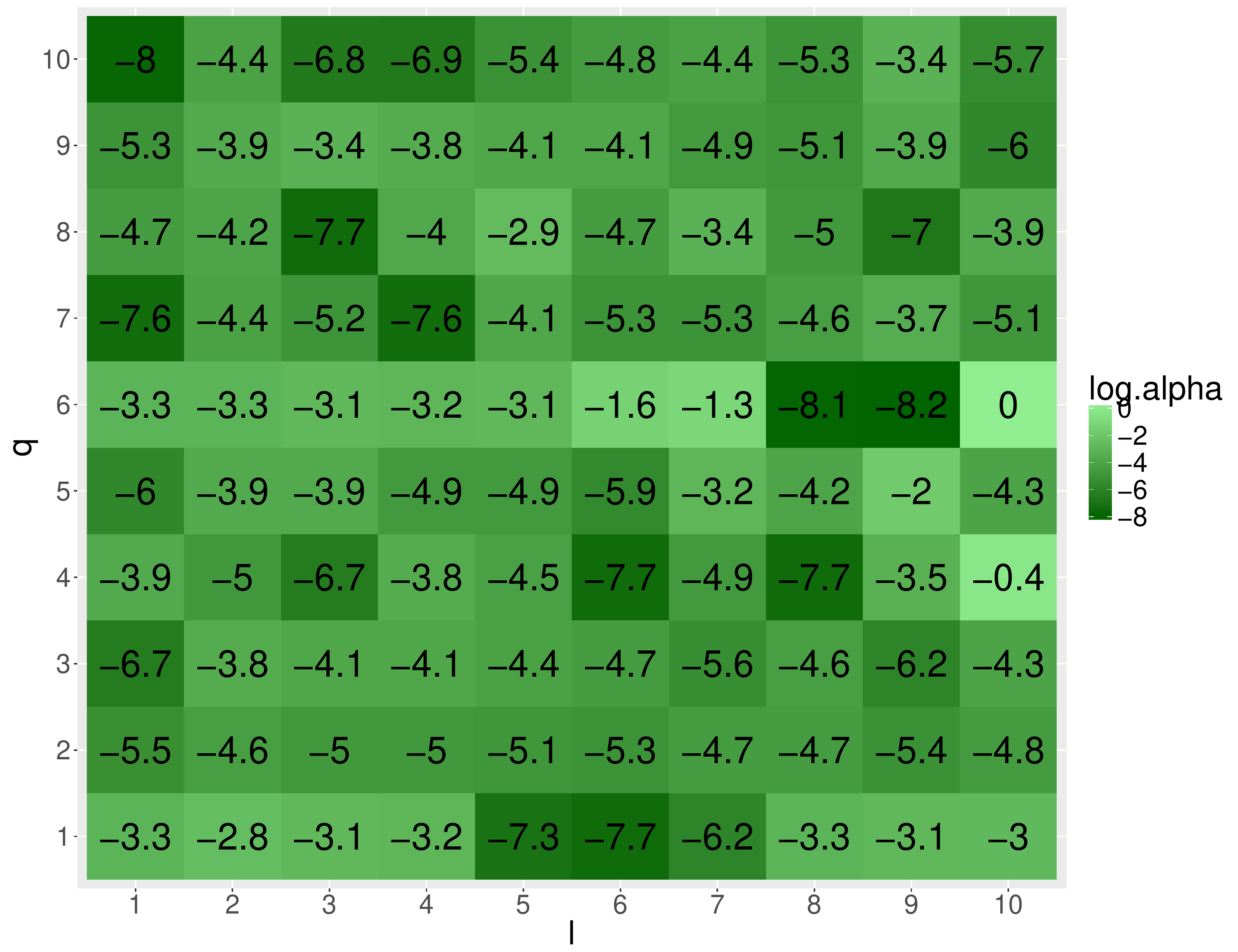}
  \caption{Enron: Logarithm of the mean values of the estimated intensities $\alpha^{(q,l)}$ in the sparse model with $Q=10$ groups.}
    \label{enron_sparse_alpha_Q10}
\end{figure}

\begin{landscape}
\begin{figure}
  \centering
  \includegraphics[width=\textwidth]{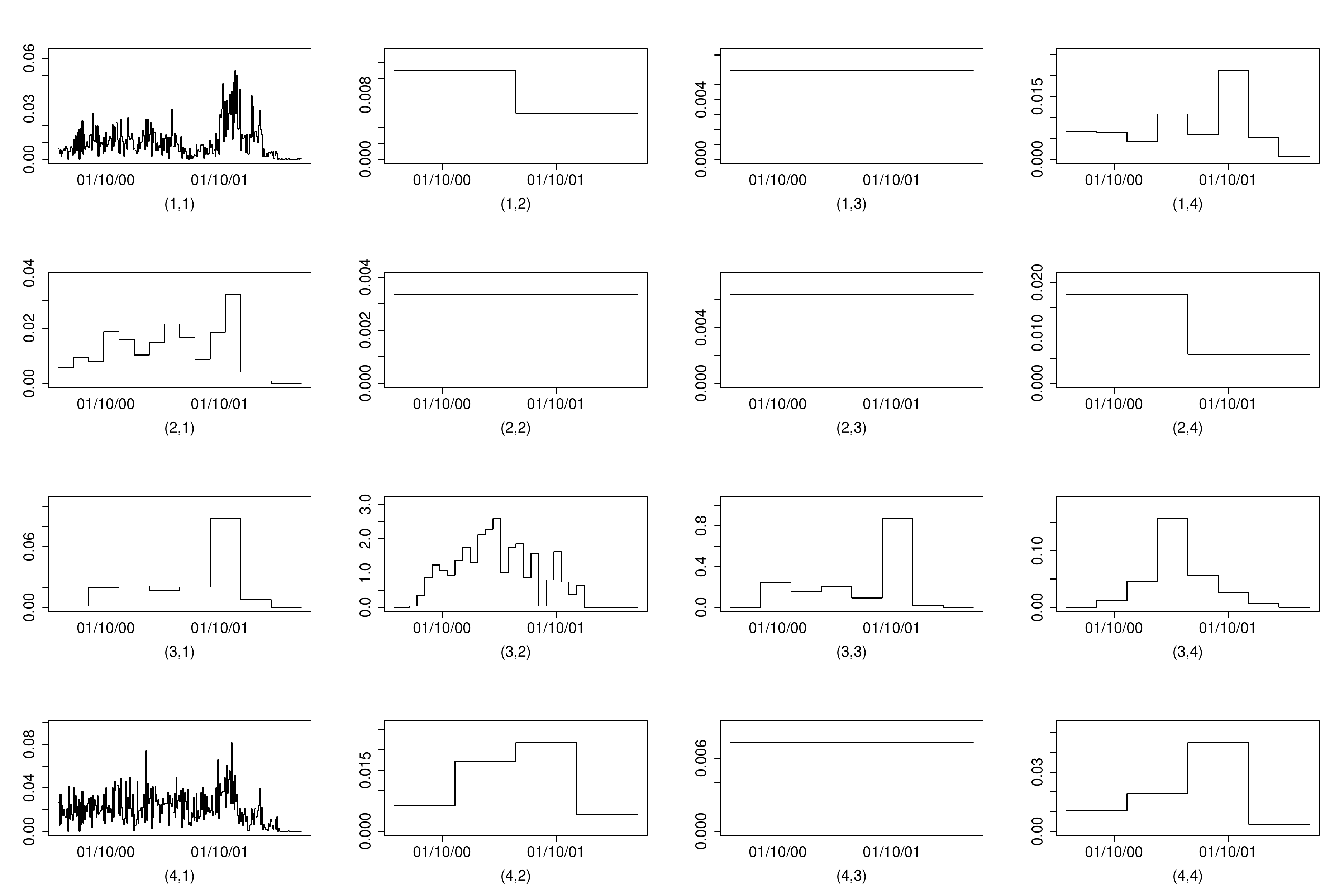}
  \caption{Enron: Estimated intensities  $\hat\alpha^{(q,l)}$ in the sparse model with $Q=4$ groups.}
    \label{enron_sparse_intens_Q4}
\end{figure}
  \end{landscape}

\begin{figure}[h]
  \centering
  \includegraphics[width=\textwidth]{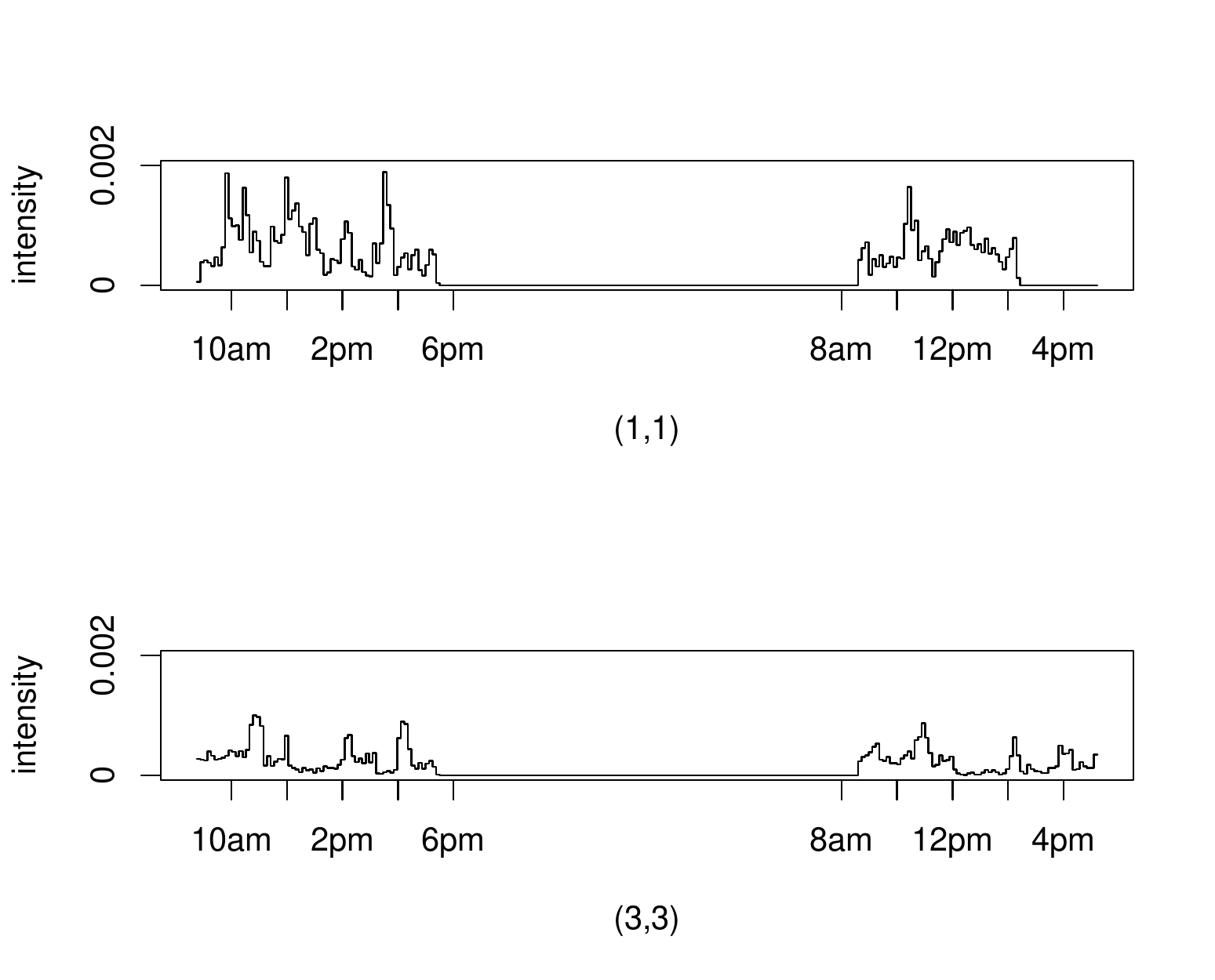}
   \caption{Primary school: Estimated  intensities in the sparse
     model. Example of group 1 (composed of class 4A and 6 pupils of
     class 4B) and group 3 (composed of the entire class 1A, with 17 pupils of class 4B
     and pupils from almost all other classes).}
  \label{primaryschool_sparse_class_4B_1A}
\end{figure}


\end{document}